\documentclass[final]{dmtcs-episciences}

\usepackage[utf8]{inputenc}

\usepackage[T1]{fontenc}
\usepackage[english]{babel}

\usepackage{amsmath}
\usepackage{amsfonts}
\usepackage{amssymb}
\usepackage{amsthm}

\usepackage{enumerate}

\newtheorem{thm}{Theorem}[section]
\newtheorem{lem}[thm]{Lemma}
\newtheorem{cor}[thm]{Corollary}
\newtheorem{propo}[thm]{Proposition}

\newtheoremstyle{problem}{}{}{}{0pt}{}{}{0pt}{}
\theoremstyle{problem}
\newtheorem*{prb}{}

\usepackage{multicol}

\def\p{{\sf P}}
\def\np{{\sf NP}}

\def\npc{{\sf NP}-complete}
\def\npcn{{\sf NP}-completeness}

\def\nphn{{\sf NP}-hardness}
\hypersetup{hidelinks}

\usepackage[numbers, sort, square]{natbib}

\author{Andreas Darmann\affiliationmark{1}
  \and Janosch D{\"o}cker\affiliationmark{2,3}\thanks{The second author thanks the New Zealand Marsden Fund for their financial support.}
  \and Britta Dorn\affiliationmark{2}}

\title{An even simpler hard variant of Not-All-Equal~3-SAT}

\affiliation{
  Department of Operations and Information Systems, University of Graz, Graz, Austria\\
  Department of Computer Science, University of T\"ubingen, T\"ubingen, Germany\\
  School of Computer Science, University of Auckland, Auckland, New Zealand}

\begin{document}

\publicationdata{vol. 28:3}{2026}{24}{10.46298/dmtcs.14908}{2024-12-06; 2024-12-06; 2026-11-13}{2026-11-18}

\setlength{\emergencystretch}{1em}
\maketitle

\begin{abstract}
 We investigate \textsc{Positive  Not-All-Equal~$3$-Sat}  under a strong structural restriction and delineate a sharp tractability threshold.  In particular,  we consider the restriction of  \textsc{Positive  Not-All-Equal~$3$-Sat}  to instances that simultaneously satisfy the following properties: 
\begin{enumerate}[(i)]
\item  the clauses can be decomposed into a disjoint union of $k$ partitions of the variable set,   and
\item  the collection of clauses is \textit{linear}, that is, each pair of distinct clauses shares at most one variable.  
\end{enumerate}

 We  show that, under these conditions, \textsc{Positive  Not-All-Equal~$3$-Sat} remains \npc\ for any  fixed $k \geq 4$.  
On the positive side, we prove that for $k \leq 3$, the problem is in \p\ and, notably, all such instances are not-all-equal-satisfiable.
Our results do not only show robustness of \npcn\ despite the strong structure imposed but also have implications for hypergraph coloring.

In particular, by the well-known correspondence between \textsc{Positive Not-All-Equal~$3$-Sat} and 
 hypergraph coloring, we obtain the following corollary of our results: For $k\geq 4$, \textsc{Bicolorability} is \npc\ for linear 3-uniform $k$-regular hypergraphs even if the edges are given as a decomposition into $k$ perfect matchings; with the same restrictions, for $k \leq 3$, \textsc{Bicolorability} is in \p, and  all such hypergraphs are bicolorable.

 In addition, we show that loosening the restriction to $3$-clauses by also allowing $2$-clauses shifts the above-mentioned tractability threshold. Moreover, we
deduce 
that every instance of \textsc{Positive Planar Not-All-Equal Sat} with at least three distinct variables per clause is nae-satisfiable. 
Besides,   for  \textsc{Linear Not-All-Equal~$3$-Sat} we derive a computational complexity dichotomy with respect to the number of variable appearances.

\end{abstract}

\section{Introduction}

\textsc{Not-All-Equal~$3$-Sat} is a notable variant of the classic \textsc{$3$-Sat} problem and frequently used as a base problem for polynomial-time reductions. Its \npcn\ has been shown by Schaefer~\cite{schaefer78} in the setting where each clause consists of at most three literals---in fact, he proves a more general result, known today as \textit{Schaefer's dichotomy theorem}, that implies \nphn\ of several interesting variants of the generalized Boolean satisfiability problem.  

Let $\mathcal{C}$ be a collection of clauses formed over a set of variables $V$. In the following, in an instance of \textsc{Not-All-Equal~$3$-Sat}, each clause consists of exactly three distinct literals, and the question is: Can we assign truth values to the variables in such a way that, for each clause, not all literals have the same truth value? In contrast to \textsc{$3$-Sat}, the decision problem corresponding to \textsc{Not-All-Equal~$3$-Sat} does not turn into a trivial problem in the absence of negations. Indeed, referring to this variant as \textsc{Positive Not-All-Equal~$3$-Sat}, it  was shown by Darmann and Döcker~\cite{DarDoe2020} that \textsc{Positive Not-All-Equal~$3$-Sat} remains \npc\ even if both (a) the collection of clauses is \textit{linear}, that is, each pair of distinct clauses shares at most one variable, and (b) each variable appears in exactly four clauses. We point out that our main result presented in this paper  imposes a stronger structural constraint on the clause family. Earlier work by Porschen et al.~\cite{porschen14} obtained \npcn\ of \textsc{Positive Not-All-Equal~$3$-Sat} for linear instances without the variable bound in (b). The bound stated in (b) is, in fact, the best possible one as \textsc{Positive Not-All-Equal~$3$-Sat} can be solved in polynomial time if each variable appears in three clauses or fewer (see Ferreira Antunes Filho~\cite[Thm.\ 3.3.2]{Filho}).
However, if the condition on the clause size is relaxed to also allow clauses with two variables, \textsc{Positive Not-All-Equal~Sat} remains \npc\ if each variable appears in at most three clauses as shown by Kratochv\'{i}l and Tuza~\cite[Thm.\ 2.1]{kratochvil02}. Dehghan et al.~\cite{dehghan15} showed that the same\footnote{The original construction of Dehghan et al.~\cite{dehghan15} contains duplications of some $2$-clauses. However, these duplications can be replaced by newly introduced  $2$-clauses without affecting nae-satisfiability, such that the resulting instance contains no duplicated clauses.} is true if each variable appears in \textit{exactly} three clauses.
A notational remark: \textsc{Positive Not-All-Equal~Sat} is also known as \textsc{Monotone Not-All-Equal~Sat}. Both names refer to the absence of negations, that is, each clause is a disjunction of unnegated variables.

We now turn to the variant of \textsc{Positive Not-All-Equal~$3$-Sat} that is the focus of this article. We introduce it first by using an example, and we defer the formal definitions to Section~\ref{sec:prelim}. Consider the following set of 9 variables:
\[
V := \{a,\, b,\, c,\, d,\, e,\, f,\, g,\, h,\, i\}.
\]
First, we partition $V$ into subsets of size 3. For example,
\[
C_1 := \{\{a,\, b,\, c\},\, \{d,\, e,\, f\},\, \{g,\, h,\, i\}\}.
\]
The collection~$C_1$ of subsets can be viewed as the Boolean formula
\[
(a \vee b \vee c) \wedge (d \vee e \vee f) \wedge (g \vee h \vee i).
\]
By construction, $C_1$ is an instance of \textsc{Positive Not-All-Equal~$3$-Sat} in which each variable appears exactly once. As alluded to above, we need at least four appearances per variable to possibly enter the realm of \npc\ problems. In order to satisfy the necessary condition on the number of variable appearances, we introduce three additional partitions of $V$ into subsets of size~3:
\begin{align*}
C_2 &:= \{\{a,\, d,\, g\},\, \{b,\, e,\, i\},\, \{c,\, f,\, h\}\}, \\ 
C_3 &:= \{\{a,\, e,\, h\},\, \{b,\, f,\, g\},\, \{c,\, d,\, i\}\}, \\
C_4 &:= \{\{a,\, f,\, i\},\, \{b,\, d,\, h\},\, \{c,\, e,\, g\}\}.
\end{align*}
We obtain an instance of \textsc{Positive Not-All-Equal~$3$-Sat} in which each variable appears exactly four times by taking the union of $C_1$, $C_2$, $C_3$ and $C_4$. If a clause is contained in more than one partition, the involved variables appear fewer than four times in the resulting instance. Hence, if we require each variable to appear exactly four times, the partitions are necessarily pairwise disjoint. For the above example it is easily checked that $C_1$, $C_2$, $C_3$ and $C_4$ are pairwise disjoint; also note that the collection of clauses $\mathcal{C} := \bigcup_{i=1}^4 C_i$ is linear. With a little effort one can see that $\mathcal{C}$ is a no-instance (see Proposition~\ref{propo:example} for a proof).

 \paragraph{Our contribution.} 
With our main result in this paper we show that, for any fixed $k \geq 4$, \textsc{Positive Not-All-Equal~$3$-Sat} remains \npc\ if restricted to instances with a variable set $V$ and a collection of clauses $\mathcal{C}$ that satisfies both of the following properties: 
\begin{enumerate}[(i)]
\item $\mathcal{C}$ is the disjoint union of $k$ partitions of $V$ into subsets of size 3, and
\item $\mathcal{C}$ is linear.
\end{enumerate}
Since each variable appears exactly $k$ times in this setting, we improve upon  an earlier  result  that stated \npc ness of  \textsc{Positive Not-All-Equal~$3$-Sat} for linear instances even when each variable appears exactly four times~\cite{DarDoe2020}. 
Notably, condition (i) ensures that within each partition, clauses are pairwise disjoint, which enforces a strong and symmetric structure that is not present in~\cite{DarDoe2020} and which corresponds to settings where variables can be divided into different types or groups. However, in contrast to that intractability result, we show that  for $k \leq 3$ the above problem is in \p\ and, in particular, each  such  instance  is a yes-instance. 

In addition, we show that if we relax the restriction to $3$-clauses and also allow $2$-clauses, \npcn\ already holds for 3 partitions (instead of 4). This holds even under the stronger requirement that each variable appears in exactly one $3$-clause (or in exactly one $2$-clause).
Moreover, we  prove that \textsc{Linear Not-All-Equal~$3$-Sat} where every variable appears exactly $p$ times positively and exactly $q$ times negatively is \npc\ for every fixed total appearance $p+q \geq 4$.
Finally, we show that every instance of \textsc{Positive Planar Not-All-Equal Sat} with at least three distinct variables per clause is not-all-equal-satisfiable. 
 In this way, we add to the literature on the computational complexity of satisfiability problems under structural constraints such as linearity, bounded variable appearances, monotonicity, and planarity (see, e.g., \cite{BerKar2007, HavivRegevTaShma2007, Pilz2019, DarDoe2021, Hasan2022, DesboisSankurSchwarzentruber2025}).
 
 Turning to our main result, observe that the merit of our disjointness condition (i) is twofold. On the one hand, when combined with the linearity condition, it strongly restricts clause interaction by ensuring that each variable participates in exactly one clause per partition. Note that (not-all-equal-)satisfiability problems where formulas are linear or the appearance of variables is bounded (as, e.g., considered in \cite{DarDoe2020, dehghan15, kratochvil02, porschen14}) are of a similar, but structurally less restrictive, flavor. Also, one might expect such a strong restriction  to allow  computational tractability. Showing that hardness remains in this regime demonstrates robustness of \npc ness despite such strong decoupling; in addition, we pinpoint a threshold ($k= 4$) where limited clause interaction becomes sufficient for \npc ness.
On the other hand, condition (i) admits a nice hypergraph interpretation; in particular, the complexity  classification in our main result refines the complexity of \textsc{Bicolorability} under  simultaneous structural constraints.

It is well known (folklore; compare, e.g., \cite{GarJoh1979}) that \textsc{Positive Not-All-Equal~$3$-Sat} is equivalent to \textsc{Set Splitting}, and to \textsc{Bicolorability} for 3-uniform hypergraphs, that is, deciding whether the vertices of  a given 3-uniform hypergraph $H$ can be colored with 2 colors such that no edge is monochromatic (in which case we call $H$ bipartite):  In fact, mapping each clause $\{x,y,z\}$ to a 3-edge $\{x,y,z\}$ and reading truth values as colors yields a bijection between not-all-equal-satisfying assignments and proper bicolorings. Hence, our results for \textsc{Positive Not-All-Equal~$3$-Sat} imply the following: For any fixed $k \geq 4$, \textsc{Bicolorability} for linear 3-uniform $k$-regular hypergraphs, where the edges are given as a decomposition into $k$ perfect matchings, is \npc; for $k \leq 3$, each such hypergraph is bicolorable. For each such hypergraph $H$ on $n$ vertices, the matching number (i.e., the maximum size of a matching) is $\nu(H) = \frac{n}{3}$. For hypergraphs with bounded matching number, that is, $\nu(H) \leq s$ for some fixed integer $s$, Li and Spirkl~\cite[Thm.\ 2.1]{Li2023} showed that \textsc{Bicolorability} for hypergraphs with edge size at most 3 can be decided in polynomial time. However, there are also families of hypergraphs with bounded matching number for which \textsc{Bicolorability} is \npc; see the dichotomies obtained by Li and Spirkl~\cite[Thms.\ 1.4 and 1.5]{Li2023}.

This work is structured as follows. In Section~\ref{sec:prelim} we present the preliminaries including formal problem statements and   the central reasoning steps employed throughout this paper. The   main results of this paper concerning the computational complexity status of  the above-mentioned restricted variants of \textsc{Positive Not-All-Equal~$3$-Sat} are presented in Section~\ref{sec:main-results}. In Section~\ref{sec:additional} we provide additional results for three different settings that, in contrast to the restrictions considered in Section~\ref{sec:main-results}, consider non-monotone instances, planar instances, or contain clauses of size 2 and 3. Finally, we conclude the paper with some remarks and an open question for future research (Section~\ref{sec:conclusion}).

\section{Preliminaries}\label{sec:prelim}

Let $V := \{x_1, x_2, \ldots, x_n\}$ be a set of \emph{variables}. A \emph{literal} is a variable $x_i$ or its negation $\overline{x_i}$; the set of literals corresponding to $V$ is $\mathcal{L}_V = \{x_i, \overline{x_i} \mid 1 \leq i \leq n\}$. A \emph{clause} is a subset of $\mathcal{L}_V$ and is called \emph{positive} if it is a subset of $V$. If a clause contains exactly $k$ distinct literals, we call it a \emph{$k$-clause}. A \emph{collection of clauses}, also called a \emph{formula}, is a  set $\mathcal{C} := \{c_1, c_2, \ldots, c_m\}$ of clauses. A collection of clauses is \emph{linear} if and only if each pair of distinct clauses shares at most one variable. A \emph{truth assignment} is a mapping $\beta\colon V \rightarrow \{T, F\}$, where $T$ and $F$ represent the truth values true and false.  We say that a truth assignment $\beta'\colon V' \rightarrow \{T,F\}$ \emph{extends} a truth assignment $\beta\colon V \rightarrow \{T, F\}$ if and only if $V\subseteq V'$ and $\beta'(x_i) = \beta(x_i)$ for each $x_i \in V$. A truth assignment $\beta$ \emph{satisfies} a literal $x_i$ if $\beta(x_i)=T$ and a literal $\overline{x_i}$ if $\beta(x_i)=F$. If a literal is satisfied, we call it a \emph{true} literal and otherwise, we call it a \emph{false} literal. 

A truth assignment $\beta$ extends to clauses as follows: If a clause $c_j$ has at least one true literal under $\beta$, the clause $c_j$ is \emph{satisfied}. 
Otherwise,  if clause $c_j$ has only false literals under $\beta$, then clause $c_j$ is not satisfied. 
Now, if the clause~$c_j$ has both a true and a false literal under $\beta$, that is, the truth values of the literals are \emph{not all equal}, we say that the truth assignment $\beta$ \emph{nae-satisfies} the clause~$c_j$. A truth assignment $\beta$ nae-satisfies a collection $\mathcal{C}$ of clauses if and only if~$\beta$ nae-satisfies each clause in $\mathcal{C}$, and $\beta$ satisfies $\mathcal{C}$ if and only if it satisfies each clause in $\mathcal{C}$.  

A \emph{partition} of a variable set $V$ is defined in the same way as a partition of arbitrary sets: Let $c_1, c_2, \ldots ,c_m$ be pairwise disjoint subsets of $V$. Then, $P = \bigcup_{j = 1}^{m}\{c_j\}$ is a partition of $V$ if and only if $\bigcup_{j = 1}^m c_j  = V$, that is, taking the union of the elements in $P$ yields $V$. We deliberately chose $c_j$ for the subsets of $V$ in this definition as we view each element of $P$ as a positive clause. Thus, a partition of $V$ can be viewed as a collection of positive clauses in which each variable appears in exactly one clause. 
Moreover, taking the union of $k$ pairwise disjoint partitions of $V$ (that is, no clause appears in more than one partition) yields a formula in which each variable appears exactly $k$ times. 

A \emph{hypergraph}  $H := (\mathcal{V}, E)$ consists of a set of variables $\mathcal{V}$ and a set of edges $E = \{e_1, e_2, \ldots, e_m\}$, where each $e_j \in E$ is a subset of $\mathcal{V}$. A hypergraph $H$ is \emph{$r$-uniform}, if and only if $|e_j| = r$ for each $j \in \{1,2,\ldots,m\}$, and \emph{$s$-regular}, if and only if each variable in $\mathcal{V}$ appears in exactly $s$ edges. A hypergraph~$H$ is \emph{linear} if and only if $|e_i \cap e_j| \leq 1$ for all distinct $i, j \in \{1,2,\ldots,m\}$. A \emph{$k$-coloring} is a mapping $\gamma\colon \mathcal{V} \rightarrow \{1,2,\ldots, k\}$, where each number represents a different color. An edge $e \in E$ is \emph{monochromatic} under $\gamma$ if and only if $\gamma$ assigns the same color to each variable in $e$. A hypergraph is \emph{$k$-colorable} if and only if there exists a $k$-coloring $\gamma$ such that no edge is monochromatic under $\gamma$. For a $2$-colorable hypergraph, we also say that the hypergraph is \emph{bicolorable} or \emph{bipartite}. A \emph{matching} $M$ is a subset of the edge set $E$ such that the elements of $M$ are pairwise disjoint. A matching $M$ is called \emph{perfect} if and only if each variable in $\mathcal{V}$ appears in some element of $M$, that is, $M$ is a partition of $\mathcal{V}$.

\subsection{Problem statements}

We start with a definition of the decision problem \textsc{Not-All-Equal Sat}, which is in~\np.

\smallskip

\begin{prb}
\noindent\textsc{Not-All-Equal Sat} (\textsc{NAE-Sat})\\
{\bf Instance.} A set $V$ of variables, and a collection $\mathcal{C}$ of clauses formed over $V$, that is, each clause in $\mathcal{C}$ is a subset of the corresponding set of literals $\mathcal{L}_V$. 

\noindent
{\bf Question.} Is there a truth assignment for $V$ that nae-satisfies $\mathcal{C}$? \end{prb}

Throughout this work, \textsc{NAE-$3$-Sat} denotes the restriction of \textsc{NAE-Sat} to instances in which each clause has exactly three distinct literals. 
In this paper, we focus on the variant of \textsc{NAE-$3$-Sat} where each clause is positive. For the sake of completeness, we give its formal definition below. 

\smallskip

\begin{prb}
\noindent\textsc{Positive NAE-$3$-Sat}\\
{\bf Instance.} A set $V$ of variables, and a collection $\mathcal{C}$ of positive $3$-clauses over $V$.  

\noindent
{\bf Question.} Is there a truth assignment for $V$ that nae-satisfies $\mathcal{C}$? \end{prb}

In Section~\ref{sub:2or3}, we will turn our attention to \textsc{NAE-$(2,3)$-Sat}, i.e., the restriction of \textsc{NAE-Sat} to instances in which each clause has either exactly two or exactly three distinct literals.

We obtain further variants of \textsc{NAE-Sat} by adding prefixes, a suffix or both to the above problem names. These additions to the problem name have in common that they restrict the base problem to instances that satisfy certain properties. For instance, if we add the prefix \textsc{Linear} to \textsc{NAE-Sat}, the collection $\mathcal{C}$ in an instance of \textsc{Linear NAE-Sat} is required to be linear.  

We use the following set of prefixes in this paper:
\[
\{\textsc{Linear},\, \textsc{Positive},\, \textsc{$k$-Disjoint},\, \textsc{Planar},\, \textsc{$3$-Connected} \}.
\]

We already alluded to the meaning of the prefix \textsc{Linear} above. The prefix \textsc{Positive} indicates that we only consider instances in which all clauses are positive, and \textsc{$k$-Disjoint} means that we require that the collection of clauses $\mathcal{C} = \bigcup_{i=1}^k C_i$ is given as a partition $\{C_1, C_2, \ldots, C_k\}$ of $\mathcal{C}$ such that, for each $i \in \{1,2,\ldots, k\}$, the collection $C_i$ has at most one appearance of each variable. The two prefixes \textsc{Planar} and \textsc{$3$-Connected} refer to properties of the bipartite \emph{incidence graph} obtained by regarding the variables and clauses as vertices and adding an edge $\{x_i, c_j\}$ if and only if the variable $x_i$ appears in the clause $c_j$; that is, the incidence graph is required to be planar and 3-connected, respectively.

The suffix $k$, for example in \textsc{Positive NAE-$3$-Sat-$k$}, means that each variable appears in at most $k$ clauses. Further, if we require each variable to appear in \emph{exactly} $k$ clauses, we use the suffix \textsc{E$k$} instead. If no suffix is used, then we do not impose a restriction on the number of variable appearances. 

Note that   for \textsc{NAE-$3$-Sat} the prefixes \textsc{Positive} and \textsc{$k$-Disjoint} combined with the suffix \textsc{E$k$} can be summarized by the property that each collection $C_i$ in the disjoint union of clauses $\mathcal{C} = \bigcup_{i=1}^k C_i$ is a partition of the variable set $V$ into subsets of size 3. Hence, we obtain the following equivalent definition of \textsc{Positive Linear $k$-Disjoint NAE-$3$-Sat-E$k$} that explicitly reflects Property (i) in the abstract; in particular, it explicitly states that each $C_i$ is required to be a partition of $V$.

\noindent\fbox{%
    \parbox{.97\textwidth}{%
        \begin{prb}
\noindent\textsc{Positive Linear $k$-Disjoint NAE-$3$-Sat-E$k$}\\
{\bf Instance.} A set $V$ of variables, and collections~$C_1, C_2, \ldots, C_k$ of positive $3$-clauses that have the following properties: 
\begin{enumerate}[(i)]
\item $C_i$ is a partition of $V$ for each $i \in \{1, \ldots, k\}$,
\item $C_i \cap C_j = \emptyset$ if $i \neq j$, and
\item $|c \cap c'| \leq 1$ for each pair of distinct clauses $c, c' \in  C_1 \cup C_2 \cup \ldots \cup C_k$.
\end{enumerate}
{\bf Question.} Is there a truth assignment for $V$ that nae-satisfies $\mathcal{C} = \bigcup_{i=1}^k C_i$? \end{prb}
}
}

The second property in the definition above ensures that $\mathcal{C} = \bigcup_{i=1}^k C_i$ is a disjoint union and, thus, that each variable appears exactly $k$ times in $\mathcal{C}$. Without this property, the union of $k$ partitions of the variable set may have fewer than $k$ appearances of some variables, which contradicts the suffix \textsc{E$k$}. The third property is the linearity condition in the absence of negations.

Finally, we define the decision problem \textsc{$\mathcal{H}$-Bicolorability} for hypergraphs. To this end, let $\mathcal{H}$ be a family of hypergraphs. 
\begin{prb}
\noindent\textsc{$\mathcal{H}$-Bicolorability}\\
{\bf Instance.} A hypergraph $H = (\mathcal{V}, E)$ such that $H \in \mathcal{H}$.  

\noindent
{\bf Question.} Is $H$ bicolorable? \end{prb}
In this paper, we consider the family $\mathcal{H}$ of linear 3-uniform $k$-regular hypergraphs such that, for each $H=(\mathcal{V}, E) \in \mathcal{H}$, the edge set $E = \bigcup_{i = 1}^k M_i$ is given as a decomposition into perfect matchings $M_1, \ldots, M_k$. Observe that for this family of hypergraphs, \textsc{$\mathcal{H}$-Bicolorability} and \textsc{Positive Linear $k$-Disjoint NAE-$3$-Sat-E$k$} are equivalent problems or, as worded in the introduction, these decision problems are different views on the same question.

\subsection{On reasoning steps: Partial assignments,  implications, and resolution}\label{sub:resolution} 

 This section gives a high-level, informal overview of the reasoning steps widely used in this paper---propagating partial assignments, collecting implied $2$-clauses, and applying resolution on those $2$-clauses; for formal background on  \emph{resolution calculus}, we refer to Schöning and Tor\'{a}n~\cite[Ch.\ 2]{schoening13}.

Let $\mathcal{C}$ be a collection of positive $3$-clauses over a set of variables $V$. 
Further, let $\beta'$ be a truth assignment for a strict subset $V'$ of $V$. Assume we would like to know if we can extend $\beta'$ by setting truth values for the variables in $V\setminus V'$ such that the resulting truth assignment nae-satisfies all clauses in $\mathcal{C}$. To this end, we may try to use $\beta'$ and the clauses in $\mathcal{C}$ to infer additional information. Let $c := \{x, y, z\}$ be a clause in $\mathcal{C}$. We distinguish the following four cases.
\begin{enumerate}[(i)]
\item Case $|c \cap V'| = 3$: Then, $\beta'$ sets all variables in $c$ to some truth value. If these truth values are not all equal, we do not infer any new information from $c$. Otherwise, $c$ is not nae-satisfied and we are done.     
\item Case $|c \cap V'| = 2$: Then, $\beta'$ sets two of the variables, say $y$ and $z$, in $c$ to some truth value. If $\beta'(y) \neq \beta'(z)$, then $c$ is already nae-satisfied. Otherwise, we deduce that $x$ must be set to the opposite truth value, i.e., the value in $\{T, F\}\setminus \{\beta'(y)\}$.    
\item Case $|c \cap V'| = 1$: Then, $\beta'$ sets one of the variables, say $z$, in $c$ to some truth value. If $\beta'(z) = T$, we deduce that $\{\bar{x}, \bar{y}\}$ must be satisfied but not necessarily nae-satisfied. If $\beta'(z) = F$, we must satisfy $\{x, y\}$. We will later explain how we use such \textit{learned} $2$-clauses to infer even more information.    
\item Case $|c \cap V'| = 0$: In this case, $\beta'$ does not assign a truth value to any variable in~$c$. Thus, we do not infer any new information from $c$. 
\end{enumerate}
Now, the information that we learn in the latter subcase of Case (ii) is straightforward. We obtain a new truth assignment $\beta''$ for $V'' := V' \cup \{x\} \subseteq V$ that we can use instead of $\beta'$ for further considerations. If $V'' = V$, we check if $\beta''$ nae-satisfies all clauses in $\mathcal{C}$ and we are done. The usefulness of Case (iii) becomes apparent if we consider multiple clauses in turn. Thereby, we obtain a collection $\mathcal{C}_\textsc{Sat}$ of $2$-clauses over $V$. Let $\mathcal{C}_\textsc{Sat}$ be non-empty. A clause such as $\{\bar{x}, \bar{y}\}$ can also be viewed as an implication $x \Rightarrow \bar{y}$ or, equivalently, as $y \Rightarrow \bar{x}$. Hence, these $2$-clauses give rise to chains of implications 
\[
\ell_1 \Rightarrow \ell_2 \Rightarrow \cdots \Rightarrow \ell_r,   
\] 
where $\ell_i$ is some variable in $V$ or its negation, i.e., a literal in $\mathcal{L}_V$.
A notable special case is $\ell_1 = \ell_r$ in which case we call the chain \emph{cyclic}. In this case, we have 
\[
\beta(\ell_1) =  \beta(\ell_2) = \cdots = \beta(\ell_r)   
\]
for each truth assignment $\beta$ that satisfies the corresponding $2$-clauses. If this cyclic chain of implications forces three variables that form a clause in $\mathcal{C}$ to the same truth value, we conclude that it is not possible to nae-satisfy this clause while also satisfying $\mathcal{C}_\textsc{Sat}$. Let us consider the non-cyclic case next. If we know that $\ell_i$ is true, for example by Case (ii), this implies that all $\ell_j$ with $i < j \leq r$ must be set true as well. Again, this may lead to a conflict with a clause in $\mathcal{C}$.  \\

Once we have inferred a collection $\mathcal{C}_\textsc{Sat}$ of $2$-clauses by Case (iii), we may use \emph{resolution} in order to infer additional $2$-clauses. For example, consider two clauses $\{x, y\}$ and $\{\bar{y}, \bar{z}\}$. Noting that the variable $y$ appears unnegated in the first clause and negated in the second clause, we can infer a $2$-clause as follows:
\[
\left (\{x, y\} \setminus \{y\}\right ) \cup \left ( \{\bar{y}, \bar{z}\} \setminus \{\bar{y}\} \right ) = \{x, \bar{z}\}.
\]
We say that we \emph{used} or \emph{performed} resolution to obtain $\{x, \bar{z}\}$. Then, adding $\{x, \bar{z}\}$ to a collection of clauses $\mathcal{C}_\textsc{Sat}$ that contains $\{x, y\}$ and $\{\bar{y}, \bar{z}\}$ has no effect on the satisfiability of $\mathcal{C}_\textsc{Sat}$. Importantly, we do not consider resolution in the setting of nae-satisfiability. We are always concerned with the \textit{satisfiability} of a collection of clauses when using resolution.

A truth assignment $\beta$ that nae-satisfies a collection of clauses exhibits a well-known symmetry in the sense that we obtain a second nae-satisfying truth assignment by interchanging the truth values true and false in $\beta$. For instance, we use this property of nae-satisfying truth assignments to simplify some case distinctions in Section~\ref{sec:main-results}.

\section{Computational complexity of \textsc{Positive Linear $k$-Disjoint NAE-$3$-Sat-E$k$}}\label{sec:main-results} 

This section is structured as follows. We begin with a proposition stating that a given collection of twelve $3$-clauses is not nae-satisfiable. In Section~\ref{sub:four-partitions} we then prove that \textsc{Positive $4$-Disjoint NAE-$3$-Sat-E$4$} is \npc. After that, we show in Section~\ref{sub:linearity} that the problem remains \npc\ even if the collection of clauses is required to be linear: in particular, we show that \textsc{Positive Linear $4$-Disjoint NAE-$3$-Sat-E$4$} is \npc. In Section~\ref{sub:general-k} we generalize this result by showing that, for any fixed integer $k \geq 4$, \textsc{Positive Linear $k$-Disjoint NAE-$3$-Sat-E$k$} is \npc. Together with the result that for $k \leq 3$, \textsc{Positive Linear $k$-Disjoint NAE-$3$-Sat-E$k$} is in \p\  (and in fact, each such instance is a yes-instance; see  Section~\ref{sub:lessthanfour}), we obtain a computational complexity dichotomy with respect to the number of partitions of the variable set (summarized in Section~\ref{sub:summary}).

\medskip

We begin with the below proposition concerning a no-instance of \textsc{Positive Linear $4$-Disjoint NAE-$3$-Sat-E$4$} referred to in the introduction.

\begin{propo}\label{propo:example}
The collection $\mathcal{C}$ of the following clauses is not nae-satisfiable.

\begin{multicols}{4} 
\begin{enumerate}
\item $\{a,\, b,\, c\}$
\item $\{d,\, e,\, f\}$ 
\item $\{g,\, h,\, i\}$

\item $\{a,\, d,\, g\}$
\item $\{b,\, e,\, i\}$
\item $\{c,\, f,\, h\}$

\item $\{a,\, e,\, h\}$
\item $\{b,\, f,\, g\}$
\item $\{c,\, d,\, i\}$

\item $\{a,\, f,\, i\}$
\item $\{b,\, d,\, h\}$
\item $\{c,\, e,\, g\}$
\end{enumerate}
\end{multicols}
\end{propo}
\begin{proof}
Assume towards a contradiction that there is a truth assignment 
\[
\beta\colon \{a,b,c,d,e,f,g,h,i\} \rightarrow \{T, F\}
\]
that nae-satisfies all of the above clauses. By symmetry of nae-satisfying assignments, it is sufficient to consider the following three cases for clause 1:

\medskip

\underline{Case $\beta(a) = T$ and $\beta(b) = \beta(c) = F$:} 
 Considering clause 4, by $\beta(a) = T$, it follows that $\beta$ must set false at least one of $d,g$ because by assumption, $\beta$ nae-satisfies clause 4. Thus, $\beta$ satisfies\footnote{Note the distinction between \textit{satisfies} and \textit{nae-satisfies} throughout the paper.} the $2$-clause $\{\bar{d}, \bar{g}\}$.
Analogously, by considering, in this order, also the clauses 12, 7, 6, 10 and 9,
we can conclude that
$\beta$ satisfies (at least) the following collection of $2$-clauses:
\[
\{
\{\bar{d}, \bar{g}\},\, \{e, g\},\, 
\{\bar{e}, \bar{h}\},\, \{f, h\},\, 
\{\bar{f}, \bar{i}\},\, \{d, i\}
\}.
\]
Using resolution (see Section~\ref{sub:resolution}), we obtain a cyclic chain of implications: $d \Rightarrow e \Rightarrow f \Rightarrow d$. Hence, $\beta(d) = \beta(e) = \beta(f)$ and clause 2 is not nae-satisfied, which is a contradiction to our assumption that $\beta$ nae-satisfies all clauses in $\mathcal{C}$.

\medskip

\underline{Case $\beta(b) = T$ and $\beta(a) = \beta(c) = F$:}  In this case,  due to the clauses 11, 7, 5, 10, 8 and 4, $\beta$ satisfies the following collection of $2$-clauses:
\[
\{
\{\bar{d}, \bar{h}\},\, \{e, h\},\, 
\{\bar{e}, \bar{i}\},\, \{f, i\},\, 
\{\bar{f}, \bar{g}\},\, \{d, g\}
\}.
\]
Using resolution, we obtain the same cyclic chain of implications and, thus, a contradiction as in the previous case.

\medskip

\underline{Case $\beta(c) = T$ and $\beta(a) = \beta(b) = F$:} Then,  considering the clauses 9, 5, 12, 8, 6 and 11, it follows that $\beta$ satisfies the following collection of $2$-clauses:
\[
\{
\{\bar{d}, \bar{i}\},\, \{e, i\},\, 
\{\bar{e}, \bar{g}\},\, \{f, g\},\, 
\{\bar{f}, \bar{h}\},\, \{d, h\}
\}.
\]
Using resolution, we obtain the same cyclic chain of implications and, thus, a contradiction as in the previous cases.\\

Thus,  no assignment  nae-satisfies the given collection of clauses.\end{proof}

\subsection{Hardness if clauses decompose into four partitions}\label{sub:four-partitions}

 We begin our computational complexity study with proving \npcn\ of \textsc{Positive $4$-Disjoint NAE-$3$-Sat-E$4$}. Along the way, we make use of the following two lemmas. 
 The first one, Lemma~\ref{lem:eq-gadget-nae-sat}, introduces a set of clauses  $\operatorname{EQ}(x_1, x_2, x_3, x_4)$ that will be used as a gadget; the second one, Lemma~\ref{lem:disjoint-to-partition}, will be used to construct instances with desired disjointness properties.

\begin{lem}\label{lem:eq-gadget-nae-sat}
Let $\operatorname{EQ}(x_1, x_2, x_3, x_4)$ be the following set of clauses, where $V_\text{aux} = \{a, b, c,d,e,f,g,h,i\}$ are newly introduced variables.  

\begin{multicols}{4} 
\begin{enumerate}
\item $\{a,\, h,\, x_2\}$
\item $\{b,\, d,\, x_4\}$ 
\item $\{c,\, e,\, i \}$
\item $\{f,\, g,\, x_1\}$

\item $\{a,\, g,\, x_4\}$
\item $\{b,\, e,\, x_3\}$
\item $\{d,\, i,\, x_1\}$
\item $\{c, f, h \}$

\item $\{a,\, f,\, x_3\}$
\item $\{b,\, i,\, x_2\}$
\item $\{e,\, h,\, x_1\}$
\item $\{c,\, d,\, g \}$

\item $\{a,\, b,\, c\}$
\item $\{d,\, e,\, f\}$
\item $\{g,\, h,\, i\}$
\end{enumerate}
\end{multicols}

Then, a truth assignment $\beta$ for $X := \{x_1, x_2, x_3, x_4\}$ can be extended to a truth assignment $\beta'$ for $X \cup V_{\text{aux}}$ such that $\beta'$ nae-satisfies all clauses in $\operatorname{EQ}(x_1, x_2, x_3, x_4)$ if and only if $\beta$ assigns the same truth value to all variables in $X$. 
\end{lem}

\begin{proof}
First,  consider  the case that $\beta$ assigns the same truth value to all variables in $X$. By symmetry of nae-satisfying assignments, we may assume that $\beta(x_i) = T$ for each $x_i \in X$. Now, assigning $T$ to all variables in $\{c, d, h\}$ and $F$ to all variables in $V_{\text{aux}}\setminus\{c, d, h\}$ nae-satisfies all clauses in $\operatorname{EQ}(x_1, x_2, x_3, x_4)$. \\

Next, we show that each truth assignment $\beta$ for $X \cup V_{\text{aux}}$ that nae-satisfies all clauses in the set $\operatorname{EQ}(x_1, x_2, x_3, x_4)$ has the property $\beta(x_1) = \beta(x_2) = \beta(x_3) = \beta(x_4)$. To this end, we use clause 13 for a case distinction. By symmetry of nae-satisfying assignments, it is sufficient to consider the following three cases:

\medskip

\underline{Case $\beta(a) = T$ and $\beta(b) = \beta(c) = F$:} Then, by the clauses 1, 10, 5, 2, 9, 6, 8 and 12 it follows that $\beta$ satisfies the following collection of $2$-clauses:
\[
\mathcal{C}_\textsc{Sat} = 
\{
\{\overline{h}, \overline{x_2}\}, \{i, x_2\},  \{\overline{g}, \overline{x_4}\}, \{d, x_4\}, \{\overline{f}, \overline{x_3}\}, \{e, x_3\},
\{f, h\}, \{d, g\}
\}.
\]
The first six clauses in $\mathcal{C}_\textsc{Sat}$ are equivalent to the following chains of implications:
\begin{align*}
h \Rightarrow \overline{x_2} \Rightarrow i, && g \Rightarrow \overline{x_4} \Rightarrow d, && f \Rightarrow \overline{x_3} \Rightarrow e.
\end{align*}
Hence, we have $g \Rightarrow d$ (by the chain with $\overline{x_4}$ in the middle) and, thus, $\beta(d) = T$ by $\{d,g\} \in \mathcal{C}_\textsc{Sat}$. Then, by clauses 7 and 14, it follows that $\beta$ also satisfies

\[
\mathcal{C}'_\textsc{Sat} = 
\{
 \{\overline{i}, \overline{x_1}\}, \{\overline{e}, \overline{f}\}
\}
\]

and we can extend the chains of implications derived from $\mathcal{C}_\textsc{Sat}$ as follows:
\begin{align*}
h \Rightarrow \overline{x_2} \Rightarrow i \Rightarrow \overline{x_1}, && g \Rightarrow \overline{x_4} \Rightarrow d, && f \Rightarrow \overline{x_3} \Rightarrow e \Rightarrow \overline{f}.
\end{align*}
Hence, $\beta(f) = F$ and, by clause 4, it follows that $\beta$ must satisfy $\{g, x_1\}$. Thus, 
\begin{equation}\label{eq:implication_chain_hig}
h \Rightarrow \overline{x_2} \Rightarrow i \Rightarrow \overline{x_1} \Rightarrow g.
\end{equation}
Since $\beta(f) = F$, we have $\beta(h) = T$ by the clause $\{f, h\} \in \mathcal{C}_\textsc{Sat}$. But then, the implication chain~\eqref{eq:implication_chain_hig} implies that $\beta(h) = \beta(i) = \beta(g) = T$ and clause 15 is not nae-satisfied. We conclude that no truth assignment $\beta$ exists with $\beta(a) = T$ and $\beta(b) = \beta(c) = F$ that nae-satisfies all clauses in $\operatorname{EQ}(x_1, x_2, x_3, x_4)$. 

\medskip

\underline{Case $\beta(b) = T$ and $\beta(a) = \beta(c) = F$:}
Then, by the clauses 2, 5, 6, 9, 10, 1, 3 and 12  it follows that $\beta$ satisfies the following collection of $2$-clauses:
\[
\mathcal{C}_\textsc{Sat} = 
\{
\{\overline{d}, \overline{x_4}\}, \{g, x_4\},  \{\overline{e}, \overline{x_3}\}, \{f, x_3\}, \{\overline{i}, \overline{x_2}\}, \{h, x_2\},
\{e, i\}, \{d, g\}
\}.
\]

The first six clauses in $\mathcal{C}_\textsc{Sat}$ are equivalent to the following chains of implications:
\begin{align*}
d \Rightarrow \overline{x_4} \Rightarrow g, && e \Rightarrow \overline{x_3} \Rightarrow f, && i \Rightarrow \overline{x_2} \Rightarrow h.
\end{align*}
Hence, we have $d \Rightarrow g$ (by the chain with $\overline{x_4}$ in the middle) and, thus, $\beta(g) = T$ by $\{d,g\} \in \mathcal{C}_\textsc{Sat}$. Then, by clauses 4 and 15, it follows that $\beta$ also satisfies

\[
\mathcal{C}'_\textsc{Sat} = 
\{
 \{\overline{f}, \overline{x_1}\}, \{\overline{h}, \overline{i}\}
\}
\] 

and we can extend the chains of implications derived from $\mathcal{C}_\textsc{Sat}$ as follows:
\begin{align*}
e \Rightarrow \overline{x_3} \Rightarrow f \Rightarrow \overline{x_1}, && i \Rightarrow \overline{x_2} \Rightarrow h \Rightarrow \overline{i}.
\end{align*}
Hence, $\beta(i) = F$ and, thus, $\beta(e) = T$ by $\{e, i\} \in \mathcal{C}_\textsc{Sat}$. Since $\beta(e) = T$, the implication chain $e \Rightarrow \overline{x_3} \Rightarrow f \Rightarrow \overline{x_1}$ implies $\beta(f) = T$ and $\beta(x_1) = F$. Now, $\beta(e) = \beta(f) = T$ implies $\beta(d) = F$ by clause 14. But then $\beta(d) = \beta(i) = \beta(x_1) = F$ and clause 7 is not nae-satisfied. Again, we conclude that there is no nae-satisfying truth assignment with $\beta(b) = T$ and $\beta(a) = \beta(c) = F$.  

\medskip

\underline{Case $\beta(c) = T$ and $\beta(a) = \beta(b) = F$:}
Then, by the clauses 3, 8, 12, 1, 2, 5, 6, 9 and 10 it follows that $\beta$ satisfies the following collection of $2$-clauses:
\[
\mathcal{C}_\textsc{Sat} = 
\{
\{\overline{e}, \overline{i}\}, \{\overline{f}, \overline{h}\}, \{\overline{d}, \overline{g}\}, \{h, x_2\}, \{d, x_4\}, \{g, x_4\}, \{e, x_3\}, \{f, x_3\}, \{i, x_2\}
\}.
\]

We obtain the implications
\begin{align}\label{eq:implication_chain_x4}
e \Rightarrow \overline{i}, && f \Rightarrow \overline{h}, && d \Rightarrow \overline{g}, && \overline{g} \Rightarrow x_4, && \overline{d} \Rightarrow x_4.
\end{align}
The latter three implications imply $\beta(x_4) = T$. Now, we need to show that $\beta(x_i) =  T$, $1 \leq i \leq 3$, is necessary to nae-satisfy all clauses in $\operatorname{EQ}(x_1, x_2, x_3, x_4)$. 

First, consider $\beta(x_1) = F$. Then, by the clauses 4, 7 and 11 it follows that $\beta$ satisfies $\{f,g\}, \{d, i\}$ and $\{e, h\}$. Together with the first three implications in~\eqref{eq:implication_chain_x4}, we get the following cyclic chain of implications:   
\[
e \Rightarrow \overline{i} \Rightarrow d \Rightarrow \overline{g} \Rightarrow f \Rightarrow \overline{h} \Rightarrow e. 
\]
But then $\beta(d) = \beta(e) = \beta(f)$ and $\beta(g) = \beta(h) = \beta(i)$. Thus, neither clause 14 nor clause 15 is nae-satisfied. 

Second, consider $\beta(x_2) = F$. Then, we have $\beta(i) = \beta(h) = T$ by the clauses $\{i, x_2\}, \{h, x_2\} \in \mathcal{C}_\textsc{Sat}$ and, thus, $\beta(e) = \beta(f) = F$ by $\{\overline{e}, \overline{i}\}, \{\overline{f}, \overline{h}\} \in \mathcal{C}_\textsc{Sat}$. Since $\beta(e) = \beta(f) = F$, we have $\beta(d) = T$ by clause 14. As shown in the previous subcase, we have $\beta(x_1) = T$ if $\beta$ nae-satisfies all clauses in $\operatorname{EQ}(x_1, x_2, x_3, x_4)$. But then $\beta(d) = \beta(i) = \beta(x_1) = T$ and clause 7 is not nae-satisfied.  

Third, consider $\beta(x_3) = F$. Then, we have $\beta(e) = \beta(f) = T$ by the clauses $\{e, x_3\}, \{f, x_3\} \in \mathcal{C}_\textsc{Sat}$ and, thus, $\beta(i) = \beta(h) = F$ by $\{\overline{e}, \overline{i}\}, \{\overline{f}, \overline{h}\} \in \mathcal{C}_\textsc{Sat}$. Since $\beta(i) = \beta(h) = F$, we have $\beta(g) = T$ by clause 15. Recalling that we have $\beta(x_1) = T$, we conclude that clause 4 is not nae-satisfied. 

We conclude that each truth assignment for $X \cup V_\text{aux}$ that nae-satisfies all clauses in $\operatorname{EQ}(x_1, x_2, x_3, x_4)$ sets all $x_i \in X$ to the same truth value. 
\end{proof}

\begin{lem}\label{lem:disjoint-to-partition}
Let $V$ be a set of variables and let $\mathcal{C}$ be a collection of positive $3$-clauses with $\mathcal{C} = \bigcup_{q = 1}^4 C_q$ such that $C_1$ is a partition of $V$ and, for $q \geq 2$, the set $C_q$ contains pairwise disjoint clauses with
\[
|V_2| = |V_3| = |V_4|, 
\]
where $V_q \subseteq V$ denotes the subset of variables that do not appear in $C_q$. 

Then, in polynomial time, we can construct a collection $\mathcal{C}' = \bigcup_{q = 1}^4 C'_q$ of positive $3$-clauses over a set of variables $V'$ such that each $C'_q$ is a partition of $V'$ and $\mathcal{C}'$ is nae-satisfiable if and only if $\mathcal{C}$ is nae-satisfiable.

Further, if $C_i \cap C_j = \emptyset$ for $1 \leq i < j \leq 4$, then $C'_i \cap C'_j = \emptyset$ for $1 \leq i < j \leq 4$. 
\end{lem}

\textbf{Example.}  Let
  $V=\{a,b,c,d,e,f,g,h,i\}$ and $\mathcal{C}=\bigcup_{q=1}^4 C_q$, where the sets of clauses are given as follows.  
   $C_1 = \{  
\{a,b,c\},\;  
\{d,e,f\},\;  
\{g,h,i\}  
\}$, 
which is a partition of $V$. Next,   
$C_2 = \{  
\{a,d,g\},\;  
\{b,e,h\}  
\}$, 
so the variables that do not appear in $C_2$ are given by  
$V_2 = V \setminus \{a,b,d,e,g,h\}  
= \{c,f,i\} $. In addition,   
$C_3 =\{\{a,e,i\},\;  
\{c,d,h\}\}$, and hence the variables that do not appear in $C_3$ are given by $V_3 = V \setminus \{a,c,d,e,h,i\}  
= \{b,f,g\}$. Finally,  $
C_4 =\{  
\{b,d,i\},\;  
\{c,e,g\}  
\}$; thus, the variables that do not appear in $C_4$ are those in $V_4$, with $V_4 = V \setminus \{b,c,d,e,g,i\}  
= \{a,f,h\}$.  
Note that $C_q$ consists of pairwise disjoint clauses, for each $q$. Moreover, we have $|V_2| = |V_3| = |V_4|$,  as required by the lemma. 

The lemma states that we hence can construct, in polynomial time, a new instance $\mathcal{C}' = \bigcup_{q=1}^4 C'_q$ over a  variable set $V'$ such that every $C'_q$ is a partition of $V'$ and $\mathcal{C}'$  is nae-satisfiable if and only if $\mathcal{C}$ is; in addition, if $C_i\cap C_j=\emptyset$ for $i\ne j$, then $C_i'\cap C_j'=\emptyset$. 
Intuitively, the construction ``fills in'' the missing variables in each $C_q$ ($q \geq 2$) using copies as well as  auxiliary variables and clause blocks so that all four families become partitions over a common ground set, without changing nae-satisfiability (details of the transformation are given in the proof of the lemma below).

\begin{proof}
First, we construct the following three copies of $\mathcal{C} = \bigcup_{q = 1}^4 C_q$:
\[
\mathcal{C}^{(r)} = \bigcup_{q = 1}^4 C_q^{(r)}, \quad r \in \{1, 2, 3\}. 
\]
Let $s$ be the positive integer with $s = |V_2| = |V_3| = |V_4|$ and let, for $r \in \{1,2,3\}$,
\[
V_2^{(r)} = \{x_1^{(r)}, \ldots, x_s^{(r)}\}, \quad V_3^{(r)} = \{y_1^{(r)}, \ldots, y_s^{(r)}\}, \quad V_4^{(r)} = \{z_1^{(r)}, \ldots, z_s^{(r)}\}
\]

be the $r$th copy of $V_2$, $V_3$ and $V_4$, respectively. Now, for each $t \in \{1,\ldots, s\}$, we introduce a set of six auxiliary variables $V_\text{aux}^t := \{a_t, b_t, c_t, d_t, e_t, f_t\}$ and the following four sets $\mathcal{S}(t, q)$ of positive $3$-clauses: 
\begin{align*}
\mathcal{S}(t, 1) &= \{
\{a_t,\, b_t,\, d_t\},\, \{c_t,\, e_t,\, f_t\}
\}, \\
\mathcal{S}(t, 2) &= 
\{
\{x_t^{(1)},\, a_t,\, d_t\},\, \{x_t^{(2)},\, b_t,\, e_t\},\, \{x_t^{(3)},\, c_t,\,f_t\}
\}, \\
\mathcal{S}(t, 3) &= 
\{
\{y_t^{(1)},\, a_t,\, e_t\},\, \{y_t^{(2)},\, b_t,\, f_t\},\, \{y_t^{(3)},\, c_t,\, d_t\}
\}, \\
\mathcal{S}(t, 4) &= 
\{
\{z_t^{(1)},\, a_t,\, f_t\},\, \{z_t^{(2)},\, b_t,\, d_t\},\, \{z_t^{(3)},\, c_t,\, e_t\}
\}.
\end{align*}   
We observe that setting each variable in~$\{a_t, b_t, c_t\}$ to $T$ and each variable in~$\{d_t, e_t, f_t\}$ to $F$ nae-satisfies these four sets for any choice of truth values for the other variables. Hence, we can introduce these four sets without affecting nae-satisfiability. We now define for~$q\in\{1, 2,3,4\}$
\begin{align*}
C'_q &:= \bigcup_{r = 1}^3 C_q^{(r)} \cup \bigcup_{t = 1}^s \mathcal{S}(t, q), \\
V' &:= \bigcup_{r = 1}^3 V^{(r)}  \cup \bigcup_{t = 1}^s V_\text{aux}^t,
\end{align*}
where $V^{(r)}$ denotes the $r$th copy of $V$. By construction, each $C'_q$ contains only positive $3$-clauses. We now show that each $C'_q$ is a partition of $V'$. 

First, in each set $\mathcal{S}(t, q)$ of positive $3$-clauses, each auxiliary variable in $V_\text{aux}^t$ appears exactly once. Hence, each auxiliary variable in $\bigcup_{t = 1}^s V_\text{aux}^t$ appears in each $C'_q$. Further, note that by construction, the sets of auxiliary variables  for $\mathcal{S}(t_1, q)$ and $\mathcal{S}(t_2, q)$ are disjoint if $t_1 \neq t_2$. Thus, in each $C'_q$  each auxiliary variable appears exactly once. 

Second, let us consider $V^{(r)}$, $r \in \{1, 2, 3\}$. These  variable sets are copies of the original variable set $V$, i.e., contain copies of the variables corresponding to~$\mathcal{C}$. Since $C_1$ is a partition of $V$, each variable in $V^{(r)}$ appears exactly once in $C'_1$ (note that $\mathcal{S}(t, 1)$ has only auxiliary variables) and, thus, $C'_1$ is a partition of $V'$. Let us now consider $C'_q$ for a fixed $q \geq 2$. For each variable $v_t \in V_q$, we placed exactly one appearance of each copy $v_t^{(1)}, v_t^{(2)}, v_t^{(3)}$ in $\mathcal{S}(t, q)$. Since $\mathcal{S}(t, q) \subseteq C'_q$ and $\mathcal{S}(t_1, q) \cap \mathcal{S}(t_2, q) = \emptyset$, we conclude that $C'_q$ is a partition of $V'$. Now, suppose that $c_j \in (C'_{q_1} \cap C'_{q_2})$ where $q_1, q_2 \in \{1,2,3,4\}$ with $q_1 \neq q_2$. Clearly, $c_j \not\in \bigcup_{t = 1}^s (\mathcal{S}(t, q_1) \cup \mathcal{S}(t, q_2))$ and, thus, $c_j \in (C_{q_1}^{(r)} \cap C_{q_2}^{(r)})$ with $r \in \{1,2,3\}$. By construction, we then have $ |C_{q_1} \cap C_{q_2}| \geq 1$. By contraposition, we conclude that if $C_{q_1} \cap C_{q_2} = \emptyset$ for $1 \leq {q_1} < {q_2} \leq 4$, then $C'_{q_1} \cap C'_{q_2}= \emptyset$ for $1 \leq {q_1} < {q_2} \leq 4$.\\

We show that $\mathcal{C}' = \bigcup_{q = 1}^4 C'_q$ is nae-satisfiable if and only if $\mathcal{C}$ is nae-satisfiable.

``$\Rightarrow$'' Let $\beta\colon V' \rightarrow \{T,F\}$ be a truth assignment that nae-satisfies $\mathcal{C}'$. Then, $\beta$ nae-satisfies the clauses $\bigcup_{q = 1}^4 C_q^{(1)}$ over the set of variables $V^{(1)}$. Hence, $\beta$ nae-satisfies a copy of $\mathcal{C}$ with renamed variables. Thus, $\mathcal{C}$ is nae-satisfiable. 

``$\Leftarrow$'' Let $\beta\colon V \rightarrow \{T,F\}$ be a truth assignment that nae-satisfies $\mathcal{C}$. We nae-satisfy $C_q^{(r)}$ for each $r \in \{1, 2, 3\}$ by setting all copies of a variable to $T$ if and only if $\beta$ sets the original variable to $T$. The remaining clauses in $\mathcal{C}'$ stem from the sets $\mathcal{S}(t, q)$ which always allow a nae-satisfying truth assignment by setting the auxiliary variables as discussed further above. Thus, $\mathcal{C}'$ is nae-satisfiable.

The proof concludes by noting  that the transformation is polynomial.\end{proof}

\begin{thm}\label{thm:positive}
\textsc{Positive $4$-Disjoint NAE-$3$-Sat-E$4$} is \npc.
\end{thm}
\begin{proof}
We provide a reduction from the \npc\ problem \textsc{Positive NAE-$3$-Sat-E$4$} (see Darmann and D\"ocker~\cite{DarDoe2020}). The proof proceeds in two steps. In the first step, given an instance $\mathcal{I} $ of \textsc{Positive NAE-$3$-Sat-E$4$}, we derive another instance  $\mathcal{I}'$ of  \textsc{Positive $4$-Disjoint NAE-$3$-Sat} which is a yes-instance if and only if $\mathcal{I}$ is a yes-instance. In the second step, from  $\mathcal{I}'$ we derive an instance $\mathcal{I}''$ of \textsc{Positive $4$-Disjoint NAE-$3$-Sat-E$4$} which is a yes-instance if and only if $\mathcal{I}'$---and thus instance $\mathcal{I}$---is a yes-instance.

Given an instance  $\mathcal{I} = (V, \mathcal{C})$ of \textsc{Positive NAE-$3$-Sat-E$4$}, for each variable $x \in V$ we introduce a set of four new variables $\{x_1, x_2, x_3, x_4\}$---one for each appearance---and replace the $i$th appearance of $x$ in the collection of clauses $\mathcal{C} = \{c_j \mid 1 \leq j\leq m\}$ by $x_i$ (the order of the appearances of a variable $x$ is induced by the indices of the clauses $c_j$ that contain $x$). 
Let $\mathcal{C}'_{\mathcal{I}}$ denote the resulting collection of clauses. 
Now, in order to obtain  instance $\mathcal{I}' = (V', \mathcal{C}')$,  we define  
\[
\mathcal{C}' := \mathcal{C}'_{\mathcal{I}} \cup \bigcup_{x \in V} \operatorname{EQ}(x_1, x_2, x_3, x_4),
\]
where $\operatorname{EQ}(x_1, x_2, x_3, x_4)$ is the gadget defined in Lemma~\ref{lem:eq-gadget-nae-sat} and each gadget has its own set of auxiliary variables denoted by $V_\text{aux}^x$. The set of variables in instance $\mathcal{I}'$ is
\[
V' = \bigcup_{x \in V} \left ( \{x_1, x_2, x_3, x_4\} \cup V_\text{aux}^x \right ). 
\]
By Lemma~\ref{lem:eq-gadget-nae-sat} it follows that $\mathcal{I}$ is a yes-instance if and only if $\mathcal{I}'$ is a yes-instance.

Next, we show that the clauses in $\mathcal{C}'$ can be decomposed into four sets $C'_1$, $C'_2$, $C'_3$ and $C'_4$ such that each set consists of pairwise disjoint clauses. This is the reason why we presented the clauses of the gadget~$\operatorname{EQ}(x_1, x_2, x_3, x_4)$ in four columns. Note that each column contains exactly one appearance of each auxiliary variable in $V_\text{aux}$ and at most one appearance of each element in $X = \{x_1, x_2, x_3, x_4\}$. Further, for $q \in \{2, 3, 4\}$, there is exactly one column in which~$x_q$ does not appear but all variables in~$X \setminus \{x_q\}$ do. We denote this column by $\operatorname{EQ}(x, q)$. The last column, which does not contain any variable in $X$, is denoted by $\operatorname{EQ}(x, 1)$. Hence, we partition the clauses in~$\operatorname{EQ}(x_1, x_2, x_3, x_4)$ as follows:
\begin{align*}
\operatorname{EQ}(x, 1) &= \{
\{a,\, b,\, c\},\, \{d,\, e,\, f\},\, \{g,\, h,\, i\}
\}, \\
\operatorname{EQ}(x, 2) &= 
\{
\{a,\, g,\, x_4\},\, \{b,\, e,\, x_3\},\, \{d,\, i,\, x_1\},\, \{c, f, h \}
\}, \\
\operatorname{EQ}(x, 3) &= \{
\{a,\, h,\, x_2\},\, \{b,\, d,\, x_4\},\, \{c,\, e,\, i \},\, \{f,\, g,\, x_1\}
\}, \\
\operatorname{EQ}(x, 4) &= 
\{
\{a,\, f,\, x_3\},\, \{b,\, i,\, x_2\},\, \{e,\, h,\, x_1\},\, \{c,\, d,\, g \}
\}.
\end{align*}    
Now, we can define the decomposition $\mathcal{C}' = \bigcup_{q=1}^4 C'_q$ as follows:
\[
  C'_q = 
  \begin{cases}
    \mathcal{C}'_{\mathcal{I}} \cup \bigcup_{x \in V} \operatorname{EQ}(x, 1), & \text{for } q = 1 \\
    \bigcup_{x \in V} \operatorname{EQ}(x, q), & \text{for } q \in \{2, 3, 4\}
  \end{cases}
\]
Note that $C'_1$ is a partition of $V'$. For a fixed $q \geq 2$, $C'_q$ consists of pairwise disjoint clauses but it is not a partition of $V'$ as $C'_q$ does not contain an appearance of any variable in $V'_q = \bigcup_{x \in V} \{x_q\}$. In fact, $V'_q$ is precisely the subset of variables $V'$ that have no appearance in $C'_q$. Since $|V'_2| = |V'_3| = |V'_4|$, the conditions in Lemma~\ref{lem:disjoint-to-partition} are satisfied and we obtain an instance $\mathcal{I}'' = (V'', \mathcal{C}'')$ with a decomposition into four sets that are partitions of $V''$ in polynomial time such that $\mathcal{I}''$ is a yes-instance if and only if $\mathcal{I}'$ is a yes-instance. Further, since $\mathcal{C}'_{\mathcal{I}} \subseteq C'_1$ contains all positive $3$-clauses that are subsets of $\{x_1, x_2, x_3, x_4 \mid x \in V\}$ and we did not duplicate a clause when partitioning the gadget~$\operatorname{EQ}(x_1, x_2, x_3, x_4)$, it now follows that $C'_{q_1} \cap C'_{q_2} = \emptyset$ for $1 \leq q_1 < q_2 \leq 4$. Then, by Lemma~\ref{lem:disjoint-to-partition}, the constructed decomposition of $\mathcal{C}''$ yields four pairwise disjoint partitions (that is, no clause appears in more than one partition). Hence, each variable appears exactly four times in $\mathcal{C}''$ and $\mathcal{I}''$ is an instance of \textsc{Positive $4$-Disjoint NAE-$3$-Sat-E$4$}.

We conclude the proof by remarking that both transformations are polynomial and, thus, the construction of $\mathcal{I}''$ from $\mathcal{I}$ such that $\mathcal{I}$ is a yes-instance if and only if $\mathcal{I}''$ is a yes-instance can be done in polynomial time.
\end{proof}

\subsection{Hardness remains when linearity constraint is added}\label{sub:linearity}

Next, we show that \npcn\ of \textsc{Positive $4$-Disjoint NAE-$3$-Sat-E$4$} holds even when restricted to instances in which the corresponding collection of clauses is linear. To this end, we state the following lemma.

\begin{lem}\label{lem:eq-gadget-linear-nae-sat}
Let $\operatorname{EQ_\text{lin}}(x_1, x_2, x_3, x_4)$ be the following set of clauses, where $V_\text{aux} = \{a, b, c, d, e,  f\}$ are new variables.  

\begin{multicols}{4} 
\begin{enumerate}
\item $\{a, d, x_3\}$
\item $\{b, e, x_2\}$
\item $\{c, f, x_4\}$

\item $\{a, c, x_1\}$
\item $\{b, d, x_4\}$
\item $\{e, f, x_3\}$

\item $\{a, e, x_4\}$
\item $\{b, f, x_1\}$ 
\item $\{c, d, x_2\}$

\item $\{a, f, x_2\}$
\item $\{b, c, x_3\}$
\item $\{d, e, x_1\}$
\end{enumerate}
\end{multicols}

Then, a truth assignment $\beta$ for $X := \{x_1, x_2, x_3, x_4\}$ can be extended to a truth assignment $\beta'$ for $X \cup V_{\text{aux}}$ such that $\beta'$ nae-satisfies all clauses in $\operatorname{EQ_\text{lin}}(x_1, x_2, x_3, x_4)$ if and only if $\beta$ assigns the same truth value to all variables in $X$. 
\end{lem}

\begin{proof}
First, consider the case that $\beta$ assigns the same truth value to all variables in $X$. By symmetry of nae-satisfying assignments, we may assume that $\beta(x_i) = T$ for each $x_i \in X$. Now, assigning the truth value $F$ to all variables in $V_{\text{aux}}$ nae-satisfies all clauses in $\operatorname{EQ_\text{lin}}(x_1, x_2, x_3, x_4)$. 

We proceed by showing that $\beta(x_1) = \beta(x_j)$ for $j \in \{2,3,4\}$ is necessary to obtain a truth assignment that nae-satisfies all clauses in $\operatorname{EQ_\text{lin}}(x_1, x_2, x_3, x_4)$.
\begin{itemize}
\item $\beta(x_1) = \beta(x_2)$: Assume towards a contradiction that $\beta(x_1) \neq \beta(x_2)$ and all clauses in the set $\operatorname{EQ_\text{lin}}(x_1, x_2, x_3, x_4)$ can be nae-satisfied. By symmetry, we may assume that $\beta(x_1) = T$ and $\beta(x_2) = F$. Then, by the clauses containing either $x_1$ or $x_2$, we have to satisfy (but not necessarily nae-satisfy)  the clauses 
\[
\{\overline{b}, \overline{f}\}, \{c, d\}, \{b, e\}, \{\overline{a}, \overline{c}\}, \{a, f\}, \{\overline{d}, \overline{e}\}
\]
which yields, by performing resolution, the following clauses:
\[
\{\overline{b}, a\}, \{\overline{a}, d\}, \{\overline{d}, b\}, \{\overline{c}, f\}, \{\overline{f}, e\}, \{\overline{e}, c\}.
\]

Note that the latter $2$-clauses form two cyclic chains of implications that can only be satisfied if $\beta'(a) = \beta'(b) = \beta'(d)$ and $\beta'(c) = \beta'(e) = \beta'(f)$. Let   
\[
\beta'(a) = \beta'(b) = \beta'(d) = b_1 \text{ and } \beta'(c) = \beta'(e) = \beta'(f) = b_2,
\]
with $b_1, b_2 \in \{T, F\}$. By the clauses 1 and 6, i.e., $\{a, d, x_3\}$ and $\{e, f, x_3\}$, we have to set $x_3$ to $\overline{b_1}$ and to $\overline{b_2}$, respectively\footnote{Formally, we set $\overline{T}=F$ and $\overline{F}=T$. }. Hence, we conclude $b_1 = b_2$, i.e., all variables in $V_\text{aux}$ have the same truth value under $\beta'$.  Since $x_1$ and $x_2$ have different truth values, in one of the clauses 2 and 4 all  variables have the same truth value no matter how we choose $b_1$. This contradicts our assumption that $\operatorname{EQ_\text{lin}}(x_1, x_2, x_3, x_4)$ can be nae-satisfied. 
\item $\beta(x_1) = \beta(x_3)$: Assume towards a contradiction that $T = \beta(x_1) \neq \beta(x_3) = F$ and all clauses in $\operatorname{EQ_\text{lin}}(x_1, x_2, x_3, x_4)$ can be nae-satisfied.
Then, by the clauses containing either $x_1$ or $x_3$, we have to satisfy the clauses
\[
\{\overline{b}, \overline{f}\}, \{a, d\}, \{e, f\}, \{\overline{a}, \overline{c}\}, \{b, c\}, \{\overline{d}, \overline{e}\}
\]
 
which is equivalent to the following cyclic chain of implications:
\[
b \Rightarrow \overline{f} \Rightarrow e \Rightarrow \overline{d} \Rightarrow a \Rightarrow \overline{c} \Rightarrow b.
\]
Hence, it must hold that  $\beta'(a) = \beta'(e) \neq \beta'(c) = \beta'(f)$. But then clauses 3 and 7 cannot be simultaneously nae-satisfied  which contradicts our assumption.     
\item $\beta(x_1) = \beta(x_4)$: Assuming $T = \beta(x_1) \neq \beta(x_4) = F$, we have to satisfy the clauses
\[
\{\overline{b}, \overline{f}\}, \{c, f\}, \{\overline{a}, \overline{c}\}, \{a, e\}, \{\overline{d}, \overline{e}\},  \{b, d\} 
\]
which is equivalent to the following cyclic chain of implications:
\[
b \Rightarrow \overline{f} \Rightarrow c \Rightarrow \overline{a} \Rightarrow e \Rightarrow \overline{d} \Rightarrow b.
\]
Thus, we get $\beta'(b) = \beta'(e) \neq \beta'(a) = \beta'(f)$. But then clauses 2 and 10 cannot be simultaneously nae-satisfied  which contradicts our assumption. 

\end{itemize}

We conclude the proof by summarizing that the condition $\beta(x_1) = \beta(x_j)$ for all $j \in \{2,3,4\}$ is indeed necessary to obtain a truth assignment that nae-satisfies all clauses in $\operatorname{EQ_\text{lin}}(x_1, x_2, x_3, x_4)$.
\end{proof}

\begin{thm}\label{thm:pos-lin-E4}
\textsc{Positive Linear $4$-Disjoint NAE-$3$-Sat-E$4$} is \npc.
\end{thm}

\begin{proof}
We show \nphn\ of \textsc{Positive Linear $4$-Disjoint NAE-$3$-Sat-E$4$} by a polynomial reduction from its version without the linearity constraint (which is \npc\ by Theorem~\ref{thm:positive}).

Let $\mathcal{I} = (V, \mathcal{C})$ be an instance of \textsc{Positive $4$-Disjoint NAE-$3$-Sat-E$4$} with $\mathcal{C} = \bigcup_{i = 1}^4 C_i$ such that $C_i$ is a partition of $V$ into subsets of size 3. We define
\[
V^{(i)} := \{x_i \mid x \in V\}, \quad i \in \{1,2,3,4\},
\]
which are copies of the variable set $V$ obtained by renaming variables, and 
\[
C^{(i)} := \{\{x_i, y_i, z_i\} \mid \{x, y, z\} \in C_i \}, \quad i \in \{1,2,3,4\},
\]
which are again obtained from their counterparts $C_1, C_2, C_3$ and $C_4$, respectively, by renaming variables. Note that the clauses in $C^{(i)}$ are formed over the set of variables $V^{(i)}$, i.e., all variables that appear in $C^{(i)}$ have the same index~$i$.
 
In order to derive an instance $\mathcal{I}' = (V', \mathcal{C}')$ of \textsc{Positive Linear $4$-Disjoint NAE-$3$-Sat-E$4$}, we define  
\[
\mathcal{C}' := \bigcup_{i = 1}^4 C^{(i)} \cup \bigcup_{x \in V} \operatorname{EQ_\text{lin}}(x_1, x_2, x_3, x_4),
\]
where $\operatorname{EQ_\text{lin}}(x_1, x_2, x_3, x_4)$ is the gadget defined in Lemma~\ref{lem:eq-gadget-linear-nae-sat} and each gadget has its own set of auxiliary variables denoted by $V_\text{aux}^x$. The set of variables in instance $\mathcal{I}'$ is
\[
V' = \bigcup_{x \in V} \left ( \{x_1, x_2, x_3, x_4\} \cup V_\text{aux}^x \right ). 
\]  

We now show that  $\mathcal{I}' = (V', \mathcal{C}')$ is indeed an instance of \textsc{Positive Linear $4$-Disjoint NAE-$3$-Sat-E$4$}, i.e., that the clause set $\mathcal{C}'$ is (1) linear and (2) consists of four pairwise disjoint partitions of $V'$. First, let us show (1). The clauses in $\bigcup_{i = 1}^4 C^{(i)}$ are pairwise disjoint and formed over the variables in $\bigcup_{x \in V} \{x_1, x_2, x_3, x_4\}$. Further, each clause in $\bigcup_{x \in V}\operatorname{EQ_\text{lin}}(x_1, x_2, x_3, x_4)$ contains exactly one variable of the set $\bigcup_{x \in V} \{x_1, x_2, x_3, x_4\}$. Hence, each clause in $\bigcup_{i = 1}^4 C^{(i)}$ has at most one variable in common with any other clause from $\mathcal{C}'$. There are no variables shared between different instances of the gadget $\operatorname{EQ_\text{lin}}(x_1, x_2, x_3, x_4)$ and it is easy to verify that each instance consists of a linear clause set. We conclude that $\mathcal{C}'$ is linear. Second, let us show (2). We partition $\operatorname{EQ_\text{lin}}(x_1, x_2, x_3, x_4)$ as follows:  
\begin{align*}
\operatorname{EQ_\text{lin}}(x, 1) &= \{
\{a, d, x_3\}, \{b, e, x_2\}, \{c, f, x_4\}
\}, \\
\operatorname{EQ_\text{lin}}(x, 2) &= 
\{
\{a, c, x_1\}, \{b, d, x_4\}, \{e, f, x_3\}
\}, \\
\operatorname{EQ_\text{lin}}(x, 3) &= \{
\{a, e, x_4\}, \{b, f, x_1\}, \{c, d, x_2\}
\}, \\
\operatorname{EQ_\text{lin}}(x, 4) &= 
\{
\{a, f, x_2\}, \{b, c, x_3\}, \{d, e, x_1\}
\}.
\end{align*}  

Next, we decompose  $\mathcal{C}' = \bigcup_{i=1}^4 C'_i$ into four partitions by setting
\[
C'_i := C^{(i)} \cup \bigcup_{x \in V} \operatorname{EQ_\text{lin}}(x, i). 
\]
Note that $C^{(i)}$ is a partition of $V^{(i)}$ and $C'_i \setminus C^{(i)}$ is a partition of $V' \setminus V^{(i)}$. Hence, $C'_i$ is a partition of $V'$ for each $i \in \{1,2,3,4\}$. Assume towards a contradiction that $C'_1$, $C'_2$, $C'_3$ and $C'_4$ are not pairwise disjoint. Then, there exists a clause $c_j \in C'_{i_1} \cap C'_{i_2}$ with $1 \leq i_1 < i_2 \leq 4$. Clearly, $c_j \not\in \bigcup_{x \in V} (\operatorname{EQ_\text{lin}}(x, i_1) \cup \operatorname{EQ_\text{lin}}(x, i_2))$ and, thus, $c_j \in (C^{(i_1)} \cap C^{(i_2)})$. Since $C^{(i_1)}$ and $C^{(i_2)}$ do not share any variables, we have $c_j \in C^{(i_1)} \cap C^{(i_2)} = \emptyset$; a contradiction. Hence, the four partitions $C'_1, \ldots, C'_4$ of $V'$ are pairwise disjoint. Since $\mathcal{C}'$ consists of four pairwise disjoint partitions of $V'$, each variable in $V'$ appears exactly four times. Moreover, it is easy to see that each clause in $\mathcal{C}'$ consists of exactly three unnegated variables. \\

It remains to show that $\mathcal{C}'$ is nae-satisfiable if and only if $\mathcal{C}$ is nae-satisfiable. 

``$\Rightarrow$'' Let $\beta'$ be a truth assignment for $V'$ that nae-satisfies $\mathcal{C}'$. By Lemma~\ref{lem:eq-gadget-linear-nae-sat}, we have $\beta'(x_1) = \beta'(x_2) = \beta'(x_3) = \beta'(x_4)$. Hence, we obtain a truth assignment $\beta$ for $V$ that nae-satisfies $\mathcal{C}$ by setting $\beta(x) = \beta'(x_1)$ for each $x \in V$. 

``$\Leftarrow$'' Let $\beta$ be a truth assignment for $V$ that nae-satisfies $\mathcal{C}$. Then, we obtain a truth assignment $\tilde{\beta}$ for $\bigcup_{x \in V}\{x_1,x_2, x_3, x_4\}$ by setting, for each $x \in V$, 
\[
\tilde{\beta}(x_i) = \beta(x), \quad i \in \{1,2,3,4\}.
\]
By construction, $\tilde{\beta}$ nae-satisfies each clause in $\bigcup_{i = 1}^4 C^{(i)}$. Note that we have $\tilde{\beta}(x_1) = \tilde{\beta}(x_2) = \tilde{\beta}(x_3) = \tilde{\beta}(x_4)$ for each $x \in V$. Hence, by multiple applications of Lemma~\ref{lem:eq-gadget-linear-nae-sat}, we can extend $\tilde{\beta}$ to a truth assignment $\beta'$ for $V'$ that nae-satisfies each clause in 
$
\bigcup_{x \in V} \operatorname{EQ_\text{lin}}(x_1, x_2, x_3, x_4)
$
by setting the auxiliary variables appropriately. Thus, $\beta'$ nae-satisfies $\bigcup_{i = 1}^4 C^{(i)} \cup \bigcup_{x \in V} \operatorname{EQ_\text{lin}}(x_1, x_2, x_3, x_4) = \mathcal{C}'$.

The transformation is polynomial, which concludes the proof.
\end{proof}

\subsection{Hardness for an increased number of partitions}\label{sub:general-k}

Now, we provide a generalization of  Theorem~\ref{thm:pos-lin-E4}, making use of the following lemma.

\begin{lem}
Let $k \geq 4$ be a positive integer. If \textsc{Positive Linear $k$-Disjoint NAE-$3$-Sat-E$k$} is \npc, then so is \textsc{Positive Linear $q$-Disjoint NAE-$3$-Sat-E$q$} with $q = k+1$. 
\end{lem}

\begin{proof}
Let $\mathcal{I} = (V, \mathcal{C})$ be an instance of \textsc{Positive Linear $k$-Disjoint NAE-$3$-Sat-E$k$} with $\mathcal{C} = \bigcup_{i = 1}^k C_i$ and $k \geq 4$.  We define
\[
V^{(j)} := \{x_j \mid x \in V\}, \quad j \in \{1,2,3\},
\]
which are copies of the variable set $V$ obtained by renaming variables, and 
\[
\mathcal{C}^{(j)} := \{\{x_j, y_j, z_j\} \mid \{x, y, z\} \in \mathcal{C} \}, \quad j \in \{1,2,3\},
\]
which are again obtained from $\mathcal{C}$ by renaming variables. Intuitively, we make three copies of $\mathcal{I}$ where each copy has its own variables. We define an instance $\mathcal{I}' = (V', \mathcal{C}')$ of \textsc{Positive Linear $q$-Disjoint NAE-$3$-Sat-E$q$} with $q = k+1$ as follows: 
\begin{align*}
V' &:= V^{(1)} \cup V^{(2)} \cup V^{(3)} \text{, and } \\
\mathcal{C}' &:= \mathcal{C}^{(1)} \cup \mathcal{C}^{(2)} \cup \mathcal{C}^{(3)} \cup \{\{x_1,\, x_2,\, x_3\} \mid x \in V\}. 
\end{align*}
For each variable $x$ of the original instance $\mathcal{I}$, the latter part of $\mathcal{C}'$ in the definition above contains the clause $\{x_1, x_2, x_3\}$ which consists of the three copies of $x$.  Now, we partition the clauses in~$\mathcal{C}'$ as follows:
\begin{align*}
\mathcal{C}'_i &=
\{\{x_1,\, y_1,\, z_1\}, \{x_2,\, y_2,\, z_2\}, \{x_3,\, y_3,\, z_3\} \mid \{x,\, y,\, z\} \in C_i \}, \quad &1 \leq i \leq k, \\
\mathcal{C}'_q &=
\{\{x_1,\, x_2,\, x_3\} \mid x \in V\}, \quad &q = k+1.
\end{align*}    
Intuitively, for $1 \leq i \leq k$, the collection $\mathcal{C}'_i$ consists of all copies of the clauses contained in $C_i$. It is straightforward to verify that, for each $i \in \{1,2,\ldots,k+1\}$, $\mathcal{C}'_i$ is a partition of $V'$ into subsets of size 3.  Since each clause in $\mathcal{C}'_q$ has exactly one variable of each set $V^{(1)}$, $V^{(2)}$ and $V^{(3)}$, it has at most one variable in common with any other clause in $\mathcal{C}'$. Further, for $1 \leq i \leq k$, the intersection of two clauses from different copies of $\mathcal{C}$ is empty. It now follows that (i) $\mathcal{C}'$ is linear and (ii) $\mathcal{C}'_1, \mathcal{C}'_2, \ldots, \mathcal{C}'_{k+1} = \mathcal{C}'_q$ are pairwise disjoint partitions of $V'$. By (ii), each variable in $V'$ appears exactly $q = k+1$ times in $\mathcal{C}'$.  We conclude that $\mathcal{I}'$ is indeed an instance of \textsc{Positive Linear $q$-Disjoint NAE-$3$-Sat-E$q$}. 

We now show that $\mathcal{I}'$ is a yes-instance if and only if $\mathcal{I}$ is a yes-instance.

``$\Rightarrow$'' Let $\beta'\colon V' \rightarrow \{T,F\}$ be a truth assignment that nae-satisfies $\mathcal{C}'$. Then, $\beta'$ nae-satisfies $\mathcal{C}^{(1)}$. Thus, we obtain a truth assignment $\beta$ that nae-satisfies $\mathcal{C}$ by setting $\beta(x) = \beta'(x_1)$ for each $x \in V$, where $x_1$ is the first copy of $x$.

``$\Leftarrow$'' Let $\beta\colon V \rightarrow \{T,F\}$ be a truth assignment that nae-satisfies $\mathcal{C}$. We now define a truth assignment $\beta'\colon V' \rightarrow \{T,F\}$ as follows. For $x \in V$, we set the copies of $x$ in $V' = V^{(1)} \cup V^{(2)} \cup V^{(3)}$ to the following truth values:  
\[
\beta'(x_1) = \beta'(x_2) = \beta(x) \text{ and } \beta'(x_3) = \overline{\beta(x)}.  
\]
Clearly, $\beta'$ nae-satisfies $\mathcal{C}^{(1)}$ and $\mathcal{C}^{(2)}$.  For the variables in $V^{(3)}$, we flipped the truth values. Thus, by symmetry of nae-satisfying assignments, $\beta'$ nae-satisfies $\mathcal{C}^{(3)}$. Since $\beta'(x_1) \neq \beta'(x_3)$ for each $x \in V$, all clauses in $\{\{x_1,\, x_2,\, x_3\} \mid x \in V\}$ are nae-satisfied. Hence, $\beta'$ nae-satisfies each clause in $\mathcal{C}'$. 

 Observing that the transformation is polynomial completes the proof.\end{proof}

Combining the above lemma with Theorem~\ref{thm:pos-lin-E4}, we directly obtain \npcn\ of \textsc{Positive Linear $k$-Disjoint NAE-$3$-Sat-E$k$} for each fixed $k\geq 4$.

\begin{thm}\label{thm:main-nae-npc}
For each fixed integer $k \geq 4$, \textsc{Positive Linear $k$-Disjoint NAE-$3$-Sat-E$k$} is \npc.
\end{thm}

By contrast, we show in the next section that for $k\in \{1,2,3\}$,  \textsc{Positive Linear $k$-Disjoint NAE-$3$-Sat-E$k$} is solvable in polynomial time, and every instance is nae-satisfiable.

\subsection{Deciding nae-satisfiability of instances with fewer than four partitions  is in~\p}\label{sub:lessthanfour}

In this section, we derive the result that for $k \in \{1,2,3\}$, \textsc{Positive Linear  $k$-Disjoint NAE-$3$-Sat-E$k$} is in \p; in particular, for  $k \in \{1,2,3\}$, each instance of \textsc{Positive Linear  $k$-Disjoint NAE-$3$-Sat-E$k$} is  a yes-instance.
Moreover, for these values of $k$---even when the linearity condition is dropped---\textsc{Positive  $k$-Disjoint NAE-$3$-Sat-E$k$} is in \p; in addition,  for $k \in \{1,2\}$, each instance of \textsc{Positive  $k$-Disjoint NAE-$3$-Sat-E$k$} is  a yes-instance.
However, it remains an interesting open question whether or not there are instances of \textsc{Positive $3$-Disjoint NAE-$3$-Sat-E$3$}  which are not nae-satisfiable.
In the following, we first present the above-mentioned  positive results without the linearity condition (which imply analogous results when linearity is imposed). Then, in Proposition~\ref{propo:3-always-yes},  we state that each instance of  \textsc{Positive Linear  $3$-Disjoint NAE-$3$-Sat-E$3$} is  a yes-instance.

\begin{propo}
For $k \in \{1, 2, 3\}$,  \textsc{Positive  $k$-Disjoint NAE-$3$-Sat-E$k$} is in \p. In addition, for  $k \in \{1, 2\}$,  each instance of \textsc{Positive  $k$-Disjoint NAE-$3$-Sat-E$k$} is a yes-instance.
\end{propo}

\begin{proof}
Let $\mathcal{I} = (V, \mathcal{C})$ be an instance of \textsc{Positive $k$-Disjoint NAE-$3$-Sat-E$k$} with $\mathcal{C} = \bigcup_{i = 1}^k C_i$ and $k \in \{1, 2, 3\}$.

First, let $k \in \{1,2\}$. Henning and Yeo~\cite[Thm.\ 4]{Henning2018} showed that each instance of \textsc{NAE-$3$-Sat-$3$} with fewer clauses than variables is nae-satisfiable and contains a \emph{free variable}. Intuitively, the truth value of a free variable does not affect nae-satisfiability of the instance and can thus be chosen freely. As each clause in $\mathcal{C}$ contains three variables and each variable in $V$ appears exactly $k < 3$ times, we have $|\mathcal{C}| < |V|$ and the result by Henning and Yeo applies. 

Second, let $k = 3$. In this case, the statement follows from the more general result that \textsc{Positive NAE-$r$-Sat-E$r$} is in \p\ for any fixed choice of $r$ (Corollary 3.3.3 in \cite{Filho}). 
\end{proof}

\begin{propo}\label{propo:3-always-yes}
Every instance of \textsc{Positive  Linear $3$-Disjoint NAE-$3$-Sat-E$3$} is a yes-instance.
\end{propo}
\begin{proof}
We derive this result via the following steps. First, we  view the given instance $\mathcal{I}=(V,\mathcal{C})$ of \textsc{Positive  Linear $3$-Disjoint NAE-$3$-Sat-E$3$} with clause set $\mathcal{C}$ and variable set $V$ as a bicoloring problem on a 3-uniform hypergraph.
Second, we consider \textit{blocking sets} in \textit{symmetric $k$-configurations}\footnote{ For the sake of completeness, we state the following definitions (see~\cite{Gropp1997,Erskine2021}). A \textit{symmetric $k$-configuration} $v_k$ is a finite incidence structure with $v$ points and
$v$ lines such that (i) each line contains $k$ points, (ii) through each point there are $k$ lines, and (iii) through two different points there is at most one line. A \textit{blocking set} in a symmetric $k$-configuration is a subset of points which intersects each line in at least $1$ and at most $k - 1$ points. The  \textit{Levi-graph} of a symmetric $k$-configuration is obtained by considering the $v$ points and 
$v$ lines as vertices, and including an edge from a point to every line containing it.} used in incidence geometry, which correspond exactly to these bicolorings, and borrow a result from this literature.
Third, we apply this result to the  biadjacency matrix  of the bipartite incidence graph  of $\mathcal{I}$ to achieve the statement of the proposition.   

For the first point, recall that, for the hypergraph $H = (P, L)$, with vertex  set $P=V$ and set of hyperedges $L=\mathcal{C}$,  
we can identify a Boolean assignment with a bicoloring of the vertices: $\mathcal{C}$ is nae-satisfiable if and only if there exists a bicoloring of $P$ such that every $\ell  \in  L$ contains both colors. 
For the second point, it is known that a  configuration  is blocking set free if and only if the corresponding hypergraph is not bicolorable (see~\cite[Lemma 3.2]{Gropp1997}). We point out that a formula in an instance $\mathcal{I}$ of  \textsc{Positive  Linear $3$-Disjoint NAE-$3$-Sat-E$3$} translates into a symmetric $3$-configuration, and a nae-satisfying assignment corresponds to a blocking set. In addition, the Levi-graph corresponds to the incidence graph of $\mathcal{I}$.  Now, it is known that a symmetric $k$-configuration is blocking set free if and only if its Levi-graph is det-extremal (i.e., the determinant of its biadjacency matrix $A$, $\det(A)$, corresponds to the \textit{permanent} of $A$, per($A$), which gives the number of perfect matchings in the graph); see~\cite{Erskine2021, Thomassen1986}. 
Thus, in order to prove the statement of the proposition, it suffices to show that the following claim holds.  

\textbf{Claim.} 
\textit{Let  $\mathcal{I}$ be an instance of \textsc{Positive  Linear $3$-Disjoint NAE-$3$-Sat-E$3$}, and let  
\(A=(a_{i,j})\) be the biadjacency matrix of the bipartite incidence graph of $\mathcal{I}$.  
Then \(\det(A) < \operatorname{per}(A)\), hence \(A\) is not det-extremal.}  
  
\textit{Proof of Claim.}  
By definition, \(a_{i,j} = 1\) if and only if variable \(x_i\) occurs in clause \(c_j\).  
Assume the clause set is the disjoint union of three partitions into 3-sets, and let $q$ denote the size of the partition (i.e., the number of clauses that form the partition). Observe that in our setting, the total number of clauses coincides with the number of variables as each variable appears exactly three times and each clause consists of exactly three variables. Thus, the biadjacency matrix is a square matrix, and the number of columns (representing the clauses) is $3q$.  In addition, without  
loss of generality, suppose that the first \(q\) columns of \(A\) correspond to the clauses (and therefore variable appearances) in  
the first partition, and the next~\(q\) columns correspond to those of the second partition.  
  
Within each partition, the clauses are pairwise disjoint and cover all variables  
exactly once.
After a suitable row/column ordering, we may assume that the columns belonging to the first  
partition are of the form  
\[  
(1,1,1,0,\ldots,0)^\top,\quad  
(0,0,0,1,1,1,0,\ldots,0)^\top,\quad \ldots,\quad  
(0,\ldots,0,1,1,1)^\top,  
\]  
so that adding up these \(q\) columns yields the all-ones vector \((1,\ldots,1)^\top\).  
The same holds for the second partition: summing its \(q\) columns also yields  
\((1,\ldots,1)^\top\). Therefore, the column set of \(A\) is linearly dependent,  
and hence \(\det(A) = 0\).  
  
  On the other hand, observe that the biadjacency matrix $A$ has the following properties:  As each clause consists of exactly three variables, each column contains exactly three $1$s; in addition, as each variable appears exactly three times, each row contains exactly three $1$s as well. Thus, 
by Voorhoeve's lower bound on permanents of   
\((0,1)\)-matrices~\cite{Voorhoeve1979} with exactly three $1$s in each
row and column, the permanent of \(A\) admits an  
exponential lower bound. In particular, for our structured, 3-regular  
bipartite incidence setting, we have \(\operatorname{per}(A) \ge 6\).  
Consequently,  
$\det(A) = 0 < 6 \le \operatorname{per}(A)$. Hence, \(A\) is not det-extremal.  \hfill $\lozenge$

Thus, the symmetric $3$-configuration associated with  any formula in an instance $\mathcal{I}$ of  \textsc{Positive  Linear $3$-Disjoint  NAE-$3$-Sat-E$3$} is not blocking set free. In turn, this means that the corresponding hypergraph is bicolorable, i.e., $\mathcal{I}$ is a yes-instance of  \textsc{Positive Linear $3$-Disjoint  NAE-$3$-Sat-E$3$}.
\end{proof}

\subsection{Summary: A dichotomy with respect to the number of partitions}\label{sub:summary}

Summarizing our findings in Section~\ref{sec:main-results} for \textsc{Positive Linear $k$-Disjoint NAE-$3$-Sat-E$k$}, we have settled its computational complexity status for any fixed positive integer $k$; the corresponding dichotomy is captured in the below theorem.

\begin{thm}
\textsc{Positive Linear $k$-Disjoint NAE-$3$-Sat-E$k$}  is 
\begin{itemize}
\item \npc\ for each fixed integer $k \geq 4$,
\item in \p\ for $k\leq 3$, with all instances being  yes-instances.
\end{itemize}
\end{thm}

Observe that this dichotomy immediately implies an analogous statement for \textsc{$\mathcal{H}$-Bicolorability} in hypergraphs as stated in Corollary~\ref{cor:hyper}.

\begin{cor}\label{cor:hyper}
Let $\mathcal{H}$ be the family of linear 3-uniform $k$-regular hypergraphs such that, for each $H=(\mathcal{V},E) \in \mathcal{H}$, a decomposition $E = \bigcup_{i = 1}^k M_i$ into perfect matchings $M_1, \ldots, M_k$ exists and is given as part of the input. Then \textsc{$\mathcal{H}$-Bicolorability} is 
\begin{itemize}
\item \npc\ for each fixed integer $k \geq 4$,
\item in \p\ for $k\leq 3$, with all instances being  yes-instances.
\end{itemize}

\end{cor}

\section{Additional results: Non-monotonicity, clauses of size 2 and 3, planarity}\label{sec:additional}

In this section, we deviate from the problem \textsc{Positive Linear $k$-Disjoint NAE-$3$-Sat-E$k$} in three different ways, leading to three different settings. 
First, we drop both our disjointness condition and the monotonicity condition, i.e., a clause may contain both a positive and a negative literal. For this setting, we can derive  a computational complexity dichotomy with respect to the number of variable appearances (Section~\ref{sub:Non-monotone}). 
Second, we allow for clauses of size $2$ as well. We show that this decreases the number of partitions at which the boundary between \p\ and \npcn\ lies  (Section~\ref{sub:2or3}).
Third, we restrict ourselves to planar instances; from a result from the literature we can derive the general result that each instance of \textsc{Positive Planar NAE-Sat} with at least three distinct variables per clause is nae-satisfiable (Section~\ref{sub:planar}).

\subsection{Dichotomy for \textsc{Linear NAE-$3$-Sat} where each variable appears $p$ times unnegated and $q$ times negated}\label{sub:Non-monotone}

Let $p$ and $q$ be fixed non-negative integers. As an interesting consequence of Theorem~\ref{thm:main-nae-npc}, we obtain for $p+q\geq 4$ a hardness result for \textsc{Linear NAE-$3$-Sat} on instances in which each variable appears exactly $p$ times unnegated and exactly $q$ times negated. In fact, we show that the following dichotomy holds.

\begin{thm}\label{thm:nae-3-sat-(p,q)}
Let $p$ and $q$ be fixed non-negative integers. \textsc{Linear NAE-$3$-Sat} for instances in which each variable appears exactly $p$ times unnegated and exactly $q$ times negated is 
\begin{itemize}
\item \npc\ if $p+q \geq 4$,
\item in \p\ if $p+q \leq 3$. 
\end{itemize}
\end{thm}

\begin{proof}
First, let $p + q = k \geq 4$. Let $\mathcal{I} = (V, \mathcal{C})$ be an instance of \textsc{Positive Linear $k$-Disjoint NAE-$3$-Sat-E$k$} with $\mathcal{C} = \bigcup_{i=1}^k C_i$ such that $C_i$ is a partition of $V$ for each $i \in \{1,2,\ldots,k\}$. Then, we obtain an instance $\mathcal{I}' = (V, \mathcal{C}')$ from $\mathcal{I}$ by replacing each variable appearance in $C_1 \cup C_2 \cup \ldots \cup C_q$ by its negation. Observe that $\mathcal{I}'$ is an instance of \textsc{Linear NAE-$3$-Sat} in which, by construction, each variable appears exactly $p$ times unnegated and exactly $q$ times negated. Since any truth assignment $\beta\colon V \rightarrow \{T,F\}$ nae-satisfies a clause $\{x, y, z\} \in \mathcal{C}$ if and only if $\beta$ nae-satisfies $\{\overline{x}, \overline{y}, \overline{z}\}$, it now follows that $\mathcal{I}$ is a yes-instance if and only if $\mathcal{I}'$ is a yes-instance. As the reduction is clearly polynomial, we conclude that \textsc{Linear NAE-$3$-Sat} is \npc\ for instances with $p+q \geq 4$.   

Second, let $p + q \leq 3$. Since \textsc{NAE-$k$-Sat-$k$} where each clause contains exactly $k$ literals has been shown to be in \p\ by Ferreira Antunes Filho~\cite[Thm.\ 3.3.1]{Filho}, it immediately follows that \textsc{Linear NAE-$3$-Sat} with $p+q \leq 3$ is in \p.
\end{proof}

\subsection{Instances with clauses of size 2 and 3}\label{sub:2or3}

In this section, we prove that \textsc{Positive Linear $3$-Disjoint NAE-Sat-E$3$} becomes \npc\  when loosening the restriction to $3$-clauses by also allowing $2$-clauses, even when each variable appears in exactly one $2$-clause or when it appears in exactly one $3$-clause. 
\begin{thm}\label{thm:2or3-one-app-in-3cl}
\textsc{Positive Linear $3$-Disjoint NAE-$(2,3)$-Sat-E$3$} is \npc\ even if each variable appears in exactly one $3$-clause.
\end{thm}

\begin{proof}
We show \nphn\ by a reduction from \textsc{Positive NAE-$3$-Sat-E$4$}, where each clause contains exactly three distinct variables and each variable appears in exactly four clauses (which is \npc~\cite[Thm.\ 1]{DarDoe2020}). Let $\mathcal{I} = (V, C)$ be an instance with a variable set $V = \{x_1, x_2, \ldots, x_n\}$ and a collection of clauses $C = \{c_1, c_2, \ldots, c_m\}$.

For each variable $x_i \in V$, we replace the $j$th appearance in $C$ by a new variable~$x_i^{(j)}$ and let $C_3^V$ denote the resulting collection of clauses. 
Further, for each $i \in \{1, \ldots, n\}$, we introduce
\[
C_{2,i}^{V,Y} := \{\{x_i^{(4)},\, y_i^{(4)}\}, \{y_i^{(4)},\, x_i^{(1)}\}\} \cup \bigcup_{j=1}^3 \{\{x_i^{(j)},\, y_i^{(j)}\}, \{y_i^{(j)},\, x_i^{(j+1)}\}\},
\]
where $y_i^{(j)}$ are new variables, and let $C_2^{V,Y} := \bigcup_{i=1}^n C_{2,i}^{V,Y}$.

 In addition, let $C_3^Y$ denote the collection of clauses that is obtained if we replace $x_i^{(j)}$ by $y_i^{(j)}$ in $C_3^V$ for $1 \leq i \leq n$ and $1 \leq j \leq 4$. We obtain an instance $\mathcal{I}' = (V', C')$ of \textsc{Positive NAE-$(2,3)$-Sat-E$3$} as follows: 
\begin{align*}
V' &= \{x_i^{(j)},\, y_i^{(j)} \mid 1 \leq i \leq n \text{ and } 1 \leq j \leq 4\}, \\
C' &= C_2^{V,Y} \cup C_3^V \cup C_3^Y. 
\end{align*}

It is not difficult to verify that $\mathcal{I}'$ is in fact an instance of \textsc{Positive Linear $3$-Disjoint NAE-$(2,3)$-Sat-E$3$}. In addition, note that each variable appears in exactly two $2$-clauses and in exactly one $3$-clause.

We show that $\mathcal{I}$ is a yes-instance if and only if $\mathcal{I}'$ is a yes-instance. 

``$\Rightarrow$'': Let $\beta \colon V \rightarrow \{T, F\}$ be a truth assignment that nae-satisfies $C$. We define a truth assignment $\beta'$ for $V'$ as follows:
\[
\beta'(x_i^{(j)}) = \beta(x_i) \text{ and } \beta'(y_i^{(j)}) \in \{T, F\} \setminus \{\beta(x_i)\}, \quad  1 \leq i \leq n \text{ and } 1 \leq j \leq 4.
\]
Since $\beta'(x_i^{(j)}) \neq \beta'(y_i^{(j)})$, each clause in $C_2^{V,Y}$ is nae-satisfied. For each $x_i \in V$, $\beta'$ sets each copy $x_i^{(j)}$ to the truth value $\beta(x_i)$. Hence, by construction, $\beta'$ nae-satisfies $C_3^V$. It remains to show that $C_3^Y$ is also nae-satisfied. Since $C_3^Y$ is a copy of $C_3^V$ with $x_i^{(j)}$ replaced by $y_i^{(j)}$ and $\beta'$ sets $x_i^{(j)}$ and $y_i^{(j)}$ to opposite truth values, the symmetry of nae-satisfying truth assignments implies that $\beta'$ nae-satisfies $C_3^Y$.   

``$\Leftarrow$'': Let $\beta' \colon V' \rightarrow \{T, F\}$ be a truth assignment that nae-satisfies $C'$. Since each clause in $C_2^{V,Y}$ can be viewed as an XOR condition on the truth values of the two contained variables, it is easy to verify that $\beta'$ satisfies the following condition:
\[
\beta'(x_i^{(1)}) = \beta'(x_i^{(2)}) = \beta'(x_i^{(3)}) = \beta'(x_i^{(4)}), \quad  1 \leq i \leq n. 
\]
We obtain a truth assignment $\beta$ for $V$ by setting $\beta(x_i) = \beta'(x_i^{(1)})$. By construction and the condition above, $\beta$ nae-satisfies $C$.    

The size of $\mathcal{I}'$ is linear in the size of $\mathcal{I}$ and the construction in polynomial time is straightforward. Hence, we conclude that the presented reduction is polynomial.
\end{proof}

\begin{thm}\label{thm:2or3-one-app-in-2cl}
\textsc{Positive Linear $3$-Disjoint NAE-$(2,3)$-Sat-E$3$} is \npc\ even if each variable appears in exactly one $2$-clause.
\end{thm}

\begin{proof}
Let $\mathcal{I} = (V, \mathcal{C})$ be an instance of \textsc{Positive Linear $4$-Disjoint NAE-$3$-Sat-E$4$} with $\mathcal{C} = \bigcup_{k=1}^4 C_k$ such that (i) $C_k$ is a partition of $V$ for each $k \in \{1,2,3,4\}$, (ii) $C_i \cap C_j = \emptyset$ if $i \neq j$, and (iii) $\mathcal{C}$ is linear. We obtain an instance of \textsc{Positive Linear $3$-Disjoint NAE-$(2,3)$-Sat-E$3$} in two steps. 

We obtain an instance $\mathcal{I}' = (V'=V, \mathcal{C}')$ of \textsc{Linear NAE-$3$-Sat} with $\mathcal{C}'= \bigcup_{k=1}^4 C_k'$ from $\mathcal{I}$ by, for $k \in \{1,2\}$, replacing each variable appearance in $C_k$ by its negation to obtain $C'_k$ and, for $k \in \{3,4\}$, setting $C_k'=C_k$. As observed in the proof of Theorem~\ref{thm:nae-3-sat-(p,q)}, $\mathcal{I}$ is a yes-instance if and only if $\mathcal{I}'$ is a yes-instance.  

Next, we construct an instance $\mathcal{I}''=(V'', \mathcal{C}'')$ of \textsc{Positive Linear $3$-Disjoint NAE-$(2,3)$-Sat-E$3$} as follows. For each $x_i \in V'$, we
\begin{enumerate}[(1)]
\item introduce two new variables $x_i^+$ and $x_i^-$ and a new $2$-clause $\{x_i^+, x_i^-\}$,
\item replace each appearance of $\overline{x_i}$ in $C'_1 \cup C'_2$ by $x_i^-$, and
\item replace each appearance of $x_i$ in $C'_3 \cup C'_4$ by $x_i^+$.
\end{enumerate}
Let $C''_k$ denote the collection of clauses obtained from $C'_k$ by replacing the variables with $x_i^+$ and $x_i^-$, respectively. Further, let 
\[
\mathcal{D}_1 = \{\{x_i^+, x_i^-\} \mid x_i \in V'\}\text{, } \mathcal{D}_2 = C''_1 \cup C''_3 \text{ and } \mathcal{D}_3 = C''_2 \cup C''_4.
\]
Then, we have $V'' = \{x_i^+, x_i^- \mid x_i \in V'\}$ and $\mathcal{C}'' = \mathcal{D}_1 \cup \mathcal{D}_2 \cup \mathcal{D}_3$. It is straightforward to verify that $\mathcal{I}''$ is indeed an instance of \textsc{Positive Linear $3$-Disjoint NAE-$(2,3)$-Sat-E$3$}, where $\mathcal{D}_1$ is a partition of $V''$ into subsets of size two and $\mathcal{D}_k$, for $k \in \{2,3\}$, is a partition of $V''$ into subsets of size three. We now show that $\mathcal{I}'$ is a yes-instance if and only if $\mathcal{I}''$ is a yes-instance.

``$\Rightarrow$'': Let $\beta'\colon V'\rightarrow \{T,F\}$ be a truth assignment that nae-satisfies $\mathcal{C}'$. Then we obtain a truth assignment $\beta''$ for $V''$ by setting $\beta''(x_i^+) = \beta'(x_i)$ and $\beta''(x_i^-) \in \{T,F\}\setminus \{\beta'(x_i)\}$ for each $x_i \in V'$. Observe that $\beta''(x_i^+) \neq \beta''(x_i^-)$ for each $x_i \in V'$ and, thus, each clause in $\mathcal{D}_1$ is nae-satisfied. Assume towards a contradiction that there exists a clause $c_j \in \mathcal{D}_2 \cup \mathcal{D}_3$ such that $c_j$ is not nae-satisfied by $\beta''$. If $c_j = \{x_r^+, x_s^+, x_t^+\}$ for indices $r,s,t \in \{1,2,\ldots, |V'|\}$, then $\beta'$ does not nae-satisfy the clause $\{x_r, x_s, x_t\}$ in $C'_3 \cup C'_4$; a contradiction. Otherwise, $c_j = \{x_r^-, x_s^-, x_t^-\}$ and, thus, $\{\overline{x_r}, \overline{x_s}, \overline{x_t}\}$ in $C'_1 \cup C'_2$ is not nae-satisfied by $\beta'$; again a contradiction. Hence, $\beta''$ nae-satisfies each clause in $\mathcal{C}''$. 

``$\Leftarrow$'': Let $\beta''\colon V''\rightarrow \{T,F\}$ be a truth assignment that nae-satisfies $\mathcal{C}''$. Then we obtain a truth assignment $\beta'$ for $V'$ by setting $\beta'(x_i) = \beta''(x_i^+)$ for each $x_i \in V'$. Assume towards a contradiction that there exists a clause $c_j$ in $\mathcal{C}'$ such that $c_j$ is not nae-satisfied by $\beta'$. If $c_j = \{x_r, x_s, x_t\} \in C'_3 \cup C'_4$, then $\beta''$ does not nae-satisfy the clause $\{x_r^+, x_s^+, x_t^+\}\in \mathcal{D}_2 \cup \mathcal{D}_3$; a contradiction. Otherwise, $c_j = \{\overline{x_r}, \overline{x_s}, \overline{x_t}\} \in C'_1 \cup C'_2$ and $\{x_r^-, x_s^-, x_t^-\}\in \mathcal{D}_2 \cup \mathcal{D}_3$ is not nae-satisfied by $\beta''$, which is again a contradiction. Hence, $\beta'$ nae-satisfies each clause in~$\mathcal{C}'$. 

As the reduction is polynomial, we conclude that \textsc{Positive Linear $3$-Disjoint NAE-$(2,3)$-Sat-E$3$} is \npc\ even if each variable appears in exactly one $2$-clause.
\end{proof}

\subsection{Remarks on planarity}\label{sub:planar} 

We conclude Section~\ref{sec:additional} with some concise remarks for the case when the planarity constraint is added, i.e., when the incidence graph is required to be planar.
It was shown by Pilz~\cite[Thm.\ 12]{Pilz2019} that every instance of \textsc{Planar Sat} in which each clause contains at least three unnegated appearances or at least three negated appearances of distinct variables is satisfiable. Noting that his proof constructs a \emph{nae-satisfying} truth assignment, we obtain the following theorem.

\begin{thm}\label{thm:planar}
Each instance of \textsc{Positive Planar NAE-Sat} with at least three distinct variables per clause is nae-satisfiable. 
\end{thm}
\begin{proof}
The construction by Pilz~\cite[Thm.\ 12]{Pilz2019} can be used without modifications to obtain a nae-satisfying truth assignment for an instance of \textsc{Positive Planar NAE-Sat} with at least three distinct variables per clause.
\end{proof}

By the above result, since we  focus on restricted variants of \textsc{Positive NAE-Sat} with exactly three distinct variables per clause in this paper, each decision problem we proved to be \npc\ in this setting becomes trivial if the planarity constraint is added as each such instance is nae-satisfiable. 

Previously, Hasan et al.~\cite[Thm.\ 1]{Hasan2022} showed that each instance of \textsc{Positive Planar $3$-Connected NAE-$3$-Sat} is nae-satisfiable. For an incidence graph to be 3-connected, it is necessary but not sufficient that each clause has at least three distinct variables. Hence, Theorem~\ref{thm:planar} identifies a strictly larger class of trivial instances of \textsc{Positive NAE-Sat} and this class properly contains all instances of \textsc{Positive Planar $3$-Connected NAE-$3$-Sat}. However, note that Hasan et al.~\cite[Thm.\ 1]{Hasan2022} present a linear-time algorithm to find a nae-satisfying truth assignment for 
\textsc{Positive Planar $3$-Connected NAE-$3$-Sat} whereas the construction by Pilz~\cite[Thm.\ 12]{Pilz2019} yields a quadratic-time algorithm.

\section{Conclusion}\label{sec:conclusion}
In this work, we have analyzed the computational complexity status of a variant of \textsc{NAE-$3$-Sat}, where the clauses form a disjoint union of $k$ partitions of the set of variables and each pair of distinct clauses shares at most one variable. Our results settle a computational complexity dichotomy for \textsc{Positive Linear $k$-Disjoint NAE-$3$-Sat-E$k$} with respect to the number $k$ of partitions: for $k\geq 4$, the problem is \npc, while for $k\leq 3$ it is in \p; in particular, for  $k\leq 3$, each instance of \textsc{Positive Linear $k$-Disjoint NAE-$3$-Sat-E$k$} is a yes-instance. As a byproduct of our results, it follows that \textsc{Linear NAE-$3$-Sat},  where each variable appears exactly $p$ times unnegated and exactly $q$ times negated, is \npc\ for each fixed value of $p+q\geq 4$.

In addition, observe that our results imply that loosening the restriction to $3$-clauses by also allowing $2$-clauses yields \npcn\ already for $3$ partitions (instead of $4$). This \npcn\ result holds in either of the cases in which each variable appears in exactly one $3$-clause or in exactly one $2$-clause.

Our results for \textsc{NAE-$3$-Sat} directly imply analogous results for the computational complexity of deciding the existence of a bicoloring for linear 3-uniform $k$-regular hypergraphs whose edges are given as a decomposition into $k$ perfect matchings: This decision problem is \npc\ for $k \geq 4$ and, otherwise, in~\p\ (in fact, in the latter case each such instance is a yes-instance).

While several results have been provided in this paper, some interesting questions still remain open as potential  directions for future research. For instance, while we conjecture  this not to be the case, to the best of our knowledge it is not known whether there are instances of  \textsc{Positive  $3$-Disjoint NAE-$3$-Sat-E$3$} which are not nae-satisfiable.

\bibliographystyle{abbrvnat}
\bibliography{mylit}

\end{document}